\documentclass[journal]{IEEEtran}

\usepackage{cite}

\usepackage[utf8]{inputenc}
\usepackage{graphicx}
\usepackage{amsfonts}
\usepackage{amsmath}
\usepackage{amssymb}
\usepackage{amsthm}
\usepackage{commath}
\usepackage{bm}
\usepackage{float}
\usepackage{subfloat}
\usepackage{subfig}
\usepackage{bbm}
\usepackage{stackengine}
\usepackage{comment}
\usepackage{booktabs}
\usepackage{multirow}
\usepackage{hyperref}
\usepackage{xcolor}
 
\usepackage[OT2,T1]{fontenc}

\stackMath

\graphicspath{{figures/}}
\DeclareGraphicsExtensions{-eps-converted-to.pdf}

\newtheorem{theorem}{Theorem}
\newtheorem{definition}{Definition}
\newtheorem{proposition}{Proposition}

\newcommand{\R}{\mathbb{R}}

\DeclareMathOperator*{\argmin}{arg\,min}

\DeclareSymbolFont{cyrletters}{OT2}{wncyr}{m}{n}
\DeclareMathSymbol{\Sha}{\mathalpha}{cyrletters}{"58}

\begin{document}

\title{Coupled Splines for Sparse Curve Fitting}

\author{Ic{\'{i}}ar~Llor{\'{e}}ns~Jover,
        Thomas~Debarre,
        Shayan~Aziznejad,
        and~Michael~Unser,~\IEEEmembership{Fellow,~IEEE}%

\thanks{The authors are with the Biomedical Imaging Group, École polytechnique fédérale de Lausanne, 1015 Lausanne, Switzerland (e-mail: iciar.llorensjover@epfl.ch, thomas.debarre@epfl.ch, shayan.aziznejad@epfl.ch, michael.unser@epfl.ch)}
\thanks{This research was supported by the European Research Council (ERC), Grant 692726-GlobalBioIm, and the Swiss National Science Foundation, Grant 200020\_184646 / 1.}
}

\maketitle

\begin{abstract}
We formulate as an inverse problem the construction of sparse parametric continuous curve models that fit a sequence of contour points. Our prior is incorporated as a regularization term that encourages rotation invariance and sparsity. We prove that an optimal solution to the inverse problem is a closed curve with spline components. We then show how to efficiently solve the task using B-splines as basis functions. We extend our problem formulation to curves made of two distinct components with complementary smoothness properties and solve it using hybrid splines. We illustrate the performance of our model on contours of different smoothness. Our experimental results show that we can faithfully reconstruct any general contour using few parameters, even in the presence of imprecisions in the measurements.
\end{abstract}

\IEEEpeerreviewmaketitle

\section{Introduction}

\IEEEPARstart{C}{ontour} tracing is a common yet rich subject in the image-processing and computer-graphics community. It has numerous applications, such as component labeling \cite{chang2003component,chang2004linear} or topological structure analysis \cite{suzuki1985topological}. The objective is to produce a contour that accurately separates two regions of a given image. This task is, however, not  without difficulties. Firstly, the edges suffer from discretization effects and intrinsic image noise. Secondly, the smoothness of the contour may need to be nonuniform, since contours are often made of smooth parts joined by sharp discontinuities.
Our goal is to extract a continuous stylized sparse parametric curve that explains a given set of ordered edge points given by possibly inaccurate two-dimensional coordinates, which is particularly relevant for raster-to-vector conversion,  \emph{i.e.} vectorization  \cite{Hori93raster}. Vectorization consists in converting raster data (e.g. pixel images) into a set of continuous functions representing the contours. This is the principle on which fonts or vector formats like pdf or svg, which allow for zooming into the objects without losing resolution, rely on. It is therefore a problem of great importance for the computer graphics community.  Our search for sparsity intimately follows Occam's razor principle of simplicity. Indeed, it heightens our probability of approaching the true curve, as many real-world signals are sparse. This is the principle on which compressed sensing hinges \cite{Donoho06, Eldar12}.

Two main approaches come to mind when thinking of contour tracing. The first one consists in joint edge detection and curve fitting. Parametric active contours are popular examples of this approach as these methods provide efficient tools for the extraction of a contour from an image, for example for point-cloud segmentation \cite{awadallah2014segmentation}. The contour consists in continuous curves that evolve through the optimization of an energy functional and iteratively approximate an image edge \cite{kass1988snakes, delgado2014snakes}. A plethora of parametric snake models can be found in the literature, mostly with model-based energy functionals \cite{Jain1998109, JacobBU04, UhlmannDCRU14, delgado2013spline, thevenaz0801}, or more recently with learning-based approaches \cite{chen2019learning,zhang2020deep}. Of particular relevance to this paper is a snake model implementation that uses basis functions and that allows for tangent control, a useful property when the smoothness of the contours is nonuniform \cite{uhlmann2016hermite}.

The second approach to contour tracing is discrete contour extraction and subsequent curve fitting. In the first approach, the entire image was used to iteratively update the contour, whereas the second approach interpolates a continuous parametric curve from a list of coordinates. This can be achieved using spline curves, which is the method of choice in computer graphics \cite{pavlidis1983curve,pottmann2002approximation,wang2006fitting,zheng2012fast}. Another popular way to tackle this is through a regularized minimization problem, the regularization enforcing prior knowledge about the curve \cite{grossman1971parametric, plass1983curve}. The method presented in this work follows the latter paradigm  by enforcing a sparsity prior. Other more classical spline-based methods enforce sparsity by simply removing knots from an initially nonsparse curve \cite{lyche1987knot,lyche1987discrete,lyche1988data,eck1995knot, bingol2019geomdl, piegl1995fundamental}. We also mention other contour-tracing algorithms based on very different techniques \cite{goshtasby2000grouping, seo2016fast}, as well as recent deep-learning based ones that are applicable to 3D contour tracing \cite{ wang2020pie,scholz2021parameterization}; see \cite{guo2020deep} for a survey.

To attain our goal, we solve a bipartite optimization problem. On one hand, we want that the candidate curve fits the existing contour points exactly. This is achieved through a data-fidelity term. On the other hand, as an infinity of curves could satisfy this fit, we have to enforce prior knowledge into our model. This prior knowledge is introduced as a regularization cost coupled with a regularization operator, the result aiming at the enforcement of desired properties. First, it is likely that the true curve has few variations, which implies that the curve has a sparse representation. Second, it is frequent that variations happen over both the horizontal and the vertical axes simultaneously. Moreover, the recovered curve should be possibly denoised. Finally, the optimization cost should not depend on a rotation of the system of coordinates. We show in this paper that these specifications lead us to a regularization cost that consists of a mixed $\left(\text{TV-}\ell_2\right)$ norm. 

In order to sparsify given data, modern regularizers include structured sparsity \cite{bachStructuredSparsity,liu2015background}, namely \cite{JacobsHDTV}, low rank regularization \cite{LowRank}, or deep prior \cite{ulyanovDeepPrior}. However, these regularizers assume a discrete setting and thus do not yield a continuous curve as a solution. In addition, the deep prior regularizer does not provide an interpretable model. In this paper, we explore the continuous setting, as we aim at the recovery of a continuous 2D curve. Similarly to \cite{DebarreFGU19}, we explore generalized total-variation (TV) regularization for continuous-domain signal reconstruction using B-splines as basis functions for an exact discretization. Finally, we choose to represent the curve with hybrid splines, which give us the tools to represent curves with nonuniform smoothness. While \cite{DebarreAU19} addressed signal reconstruction using hybrid splines, this manuscript extends the setting for the handling of curves in 2D, which calls for a new regularizer. 

Our main contribution is threefold. Firstly, we introduce a continuous rotation-invariant TV (RI-TV) norm as a regularization for the recovery of curves. It effectively reconstructs sparse parametric curves from given contour points while being robust to noise. Secondly, we prove a representer theorem according to which there exists a curve with spline components that is a global minimizer of our optimization problem. Building upon this, we propose a curve construction using B-splines, which allows us to discretize the continuous-domain problem exactly with numerical efficiency. Finally, we present the combination of such RI-TV norm with a hybrid framework to generate stylized curves with nonuniform smoothness properties. 

The remainder of this paper is organized as follows: In Section \ref{sec:tv_based_optim}, we present the continuous-domain framework of the optimization problem and introduce our representer theorem. We then show the precise implementation and resolution of this task through the introduction of B-splines in Section \ref{sec:discrete_implementation}. In Section \ref{sec:hybrid_splines}, we extend the framework to hybrid splines. Finally, we experimentally verify properties of our contributions and show applications in Section \ref{sec:experiments}.

\section{Continuous-Domain Formulation}\label{sec:tv_based_optim}

Our goal is to recover a 2D parametric curve $\textbf{r}(t) = (x(t), y(t))$ that best fits a given ordered list of points $\textbf{p}[m] = (\mathrm{p}_{x}[m], \ \mathrm{p}_{y}[m]), \ m = 0, \ldots, M-1$. Contours being closed curves, we consider the coordinate functions $x(t)$ and $y(t)$ to be periodic in $t$. Since we have $M$ data locations, it is convenient to deal with $M$-periodic functions. We consequently set $t \in \mathbb{T}_M = [0, M]$. 

Concurrently, we want to control the sparsity of the fitted curve. This can be achieved by limiting the number of the singularities in the higher-order derivatives of its components. This effectively means that $\textbf{r}(t)$ admits a sparse representation. To that end, we introduce two new elements: a differential operator $\mathrm{L}$ and the RI-TV regularization functional.  

The mathematical foundation on which this paper is built is based on Schwartz' theory of distributions \cite{schwartz1957theorie}. Henceforth, let us denote the Schwartz' space of $M$-periodic smooth functions by $\mathcal{S}(\mathbb{T}_M)$. Its topological dual, $\mathcal{S}'(\mathbb{T}_M)$ is the space of tempered distributions over the torus.

\subsection{Derivative Operators and ${\rm L}$-Splines}\label{sec:L-splines}

The first element we introduce in our problem formulation is $\mathrm{L} = \mathrm{D}^{\alpha+1}$, the derivative operator whose order $(\alpha+1)$, with $\alpha \in \mathbb{N} \setminus \{0\}$, determines the smoothness of the components of the constructed curve.

Next, we define the periodic $\mathrm{L}$-splines with respect to the operator $\mathrm{L}$. A periodic $\mathrm{L}$-spline is a function $s : \mathbb{T}_M \rightarrow \R$ that verifies that 

\begin{equation} \label{eq:L_spline}
    \mathrm{L} \{ s \}(t) = \sum_{k = 0}^{K-1} a[k] \Sha_M(t - t_k),
\end{equation}
 
\noindent where  $\Sha_M(t)=\sum_{k\in\mathbb{Z}}\delta(t-M k )\in \mathcal{S}'(\mathbb{T}_M)$ is the $M$-periodic Dirac comb, $K\in \mathbb{N} \setminus\{0\}$ is the number of knots, $a[k] \in \R$ is the amplitude of the $k$th jump, and $t_k \in \R$ are distinct knot locations. 

\subsection{RI-TV Regularization}

The second element is $\mathcal{R}$, a sparsity-promoting regularization functional with key characteristics. Firstly, for 2D curves, the minimization of $\mathcal{R}(\mathrm{L}\{\textbf{r}\})$, where $\mathrm{L}\{\textbf{r}\} = (\mathrm{L}\{x\}, \mathrm{L}\{y\})$, should enforce sparsity jointly for the two components of $\textbf{r}$. Indeed, we want $\textbf{r}$ to have few variations, and they often should occur along both components simultaneously. Secondly, if the points $\textbf{p}[m]$ are rotated by an angle $\theta$, the fitted curve $\textbf{r}(t)$ should be rotated by the same angle $\theta$. To achieve this, our regularizer $\mathcal{R}$ should be invariant to a rotation of the system of coordinates, meaning that $\mathcal{R}(\mathbf{R}_{\theta} \mathrm{L}\{\textbf{r}\}) = \mathcal{R}(\mathrm{L}\{\textbf{r}\})$, where $\mathbf{R}_{\theta}$ is a rotation matrix.  Similarly, $\mathcal{R}$ should be equivariant to isotropic scaling, meaning that there exists a function $\mathrm{A}$ such that $\mathcal{R}(\mathrm{L}\{a\textbf{r}\}) = \mathrm{A}(a) \ \mathcal{R}(\mathrm{L}\{\textbf{r}\})$ for any $a \neq 0$. We now introduce the RI-TV norm, which consists in a mixed continuous $\left(\text{TV-}\ell_2\right)$ norm and satisfies our specifications. 

\begin{definition}\label{Def:Regularization}
Let $p\in [1,+\infty]$. The $\left(\text{TV-}\ell_p\right)$ norm of any vector-valued tempered distribution ${\bf w}=\begin{bmatrix}w_1 & w_2 \end{bmatrix} \in \mathcal{S}'(\mathbb{T})^2$ is defined as 
\begin{equation}\label{Eq:Reg}
 \|{\bf w}\|_{{\rm TV}-\ell_p} \overset{\Delta}{=} \sup_{\substack{\boldsymbol{\varphi}=(\varphi_1,\varphi_2)\in \mathcal{S}(\mathbb{T}_M)^2\\ \|\boldsymbol{\varphi}\|_{q,\infty} = 1}} \left( \langle  w_1,\varphi_1\rangle + \langle w_2, \varphi_2 \rangle\right),
\end{equation}
where $q\in[1,\infty]$ is the H\"older conjugate of $p$ with $\frac{1}{p}+\frac{1}{q}=1$ and $\|\cdot\|_{q,\infty}$ is the  $(\ell_q - L_\infty)$ mixed norm, defined for any $\boldsymbol{\varphi} \in \mathcal{S}(\mathbb{T}_M)^2$ as
\begin{equation}\label{Eq:L2LinfMixed}
\|\boldsymbol{\varphi}\|_{q,\infty} \overset{\Delta}{=} \sup_{t\in\mathbb{T}_M} \|\boldsymbol{\varphi}(t)\|_q.
\end{equation}
\end{definition} 

A mixed norm similar to \eqref{Eq:Reg} was previously introduced in \cite{FernandezGrandaSuperres} in the context of the recovery of Dirac distributions. In Theorem \ref{Thm:Intuition}, we compute the $\left(\text{TV-}\ell_p\right)$ norm for two general classes of functions or distributions. 

\begin{theorem}\label{Thm:Intuition}
\begin{enumerate}
\item \label{Item:L1} For any curve ${\bf f}=(f_1, f_2)$ with absolutely integrable components $f_i  \in L_1(\mathbb{T}_M)$, $i=1,2$, we have that  
\begin{equation}
     \norm{[f_1 \ \ f_2]}_{\mathrm{TV}-\ell_p} = \int_0^M \|{\bf f}(t)\|_p {\rm d} t.
\end{equation}

\item \label{Item:Dirac} Let ${\bf w}=(w_1,w_2)$ be a vector-valued distribution of the form ${\bf w}=  \sum_{k =1}^K {\bf a}[k] \Sha_M(\cdot - t_k)$ with ${\bf a}[k] \in \mathbb{R}^2, k=0,\ldots,K-1$. Then, we have that 

\begin{equation}
    \norm{[w_1 \ \ w_2]}_{\mathrm{TV}-\ell_p} = \sum_{k=0}^{K-1} \|{\bf a}[k]\|_p.
\end{equation}
\end{enumerate}
\end{theorem}

The proof of Theorem \ref{Thm:Intuition} can be found in Appendix \ref{sec:appendix_norm_equivalence}. The outer TV norm promotes sparsity, as it is the continuous counterpart of the $\ell_1$ norm \cite{DebarreFGU19}. The inner $\ell_p$ norm in Item \ref{Item:L1} induces a coupling of the $f_1$ and $f_2$ components. Indeed, it first aggregates the $f_1$ and $f_2$ curve components, which the outer TV norm then jointly sparsifies. This is true of any $\ell_p$ norm for $p \neq 1$. For $p=1$, the components are no longer coupled due to the separability of the norm. For $p=2$ and any curve $\textbf{f}=(f_1, f_2)$, we set 

\begin{equation}\label{eq:tvl2_R}
    \mathcal{R}(\textbf{f}) =  \norm{[f_1 \ \ f_2]}_{\mathrm{TV}-\ell_2}.
\end{equation}
 
\begin{proposition}\label{Prop:Invariance_L2}
The $\left(\text{TV-}\ell_2\right)$ norm, notated $\mathcal{R}$, is invariant to rotation, in the sense that $\mathcal{R}(\mathbf{R}_{\theta} \mathbf{f}) =  \mathcal{R}(\mathbf{f})$, where $\mathbf{R}_{\theta}$ is a rotation matrix. Furthermore, the $\left(\text{TV-}\ell_2\right)$ norm is the only $\left(\text{TV-}\ell_p\right)$ norm that is rotation invariant.
\end{proposition}

We provide in Appendix \ref{sec:appendix_invariance} the proof of Proposition \ref{Prop:Invariance_L2}. Finally, $\mathcal{R}$ being a norm, it is scale equivariant (homogeneity property).

\subsection{Continuous-Domain Optimization Problem}

The setting we described in this section is typical of a minimization problem with two terms. The first term---the data-fidelity term---ensures that the candidate curve $\textbf{r}(t)$ is close to the points $\textbf{p}[m]$. The second term, called regularization, introduces our \emph{a priori} desiderata for the reconstructed curve. The importance of these two terms is weighted by a parameter $\lambda > 0$. The solution set of the minimization problem is

\begin{equation}\label{eq:continuous_problem}
    \mathcal{V} = \argmin_{\textbf{r} \in \mathcal{X}_{\mathrm{L}}(\mathbb{T}_M)} \left ( \sum_{m=0}^{M-1} \norm{\left.\textbf{r}(t)\right\vert_{t=m} - \textbf{p}[m]}_2^2
    + 
    \lambda \mathcal{R}(\mathrm{L}\{\textbf{r}\}) \right),
\end{equation}

\noindent where the search space $\mathcal{X}_{\rm L}$ is defined as 
\begin{align}\label{Eq:DefNative}
\mathcal{X}_{\rm L}(\mathbb{T}_M)  = \{ &{\bf r}\in  \mathcal{S}'(\mathbb{T}_M)^2: \mathcal{R}({\rm L }\{\bf r\}) <+\infty\}.
\end{align}
  
 The data-fidelity term in \eqref{eq:continuous_problem} penalizes the Euclidean distance between the sample $\left.\textbf{r}(t)\right\vert_{t=m}$ of the curve and the point $\textbf{p}[m]$ for every $m=0, \ldots, M-1$. The fact that $\textbf{r}$ is sampled uniformly along the parameter axis encourages the reconstructed curve to be parametrized by its curvilinear abcissa, promoting the arc length of $\textbf{r}(t)$ to be a linear function of the parameter $t$. The underlying assumption behind this statement is that the points $\textbf{p}[m]$ are spread approximately uniformly along the curve. This is an important assumption, since the regularization term in \eqref{eq:continuous_problem} involves the derivatives of $\textbf{r}(t)$ and thus heavily depends on the choice of parametrization. In that respect, the curvilinear abcissa is a desirable choice. In practice, it often results in rough curves being penalized heavily by our regularization, which other parametrizations may fail to achieve \cite{JacobBU04}. 
 
Our representer theorem (Theorem \ref{Thm:Rep}) specifies the form of the solution of \eqref{eq:continuous_problem}. The proof of this theorem is provided in Appendix \ref{sec:appendix_rep_theorem}.

\begin{theorem}\label{Thm:Rep}
 The  global minimizer of \eqref{eq:continuous_problem} can be achieved by a periodic ${\rm L}$-spline curve ${\bf r}^*$ with at most $K\leq 2M$ knots. Indeed, we have that 
\begin{equation}\label{eq:extreme_point}
    {\rm L}\{{\bf r}^*\} = \sum_{k=0}^{K-1} {\bf a}_k \Sha_M(\cdot -t_k)
\end{equation}
for some  distinct knot locations   $t_k\in\mathbb{T}_M$ and amplitude vectors ${\bf a}_k \in \mathbb{R}^2$.
\end{theorem}

Theorem \ref{Thm:Rep} states that the solution set \eqref{eq:continuous_problem} contains periodic $\mathrm{L}$-splines. Even though our work uses a mixed $\left(\text{TV-}\ell_2\right)$ norm as regularization, this result is reminiscent of \cite{FisherJerome1975, UnserFW2017, fageot2020tv}, which proves that inverse problems with TV regularization have spline solutions also for the recovery of signals that are periodic, with a period of integer length. 

\section{Exact Discretization} \label{sec:discrete_implementation}
 
\subsection{Polynomial B-Splines} \label{sec:b_splines}

Theorem \ref{Thm:Rep} motivates our discretization of the continuous-domain problem in \eqref{eq:continuous_problem} over the space of periodic cardinal $\mathrm{L}$-splines, \emph{i.e.} with integer knot spacing $(t_{k+1} - t_k) = 1$ in \eqref{eq:L_spline}. In the case where $\mathrm{L}$ is the multiple-order derivative operator $\mathrm{D}^{\alpha+1}$, $\mathrm{L}$-splines are piecewise polynomials of degree $\alpha$. Following \cite{UnserSplinesPerfectFit, de1978practical}, we consider symmetric $(\alpha+1)$-order B-splines $\beta^\alpha$. B-splines are the shortest functions within the space of cardinal periodic $\mathrm{L}$-splines, with a support included in $[-\frac{\alpha+1}{2}, \frac{\alpha+1}{2}]$ \cite{Deboor1972, Ron1990}. This property is particularly advantageous for numerical efficiency. We show in Figure \ref{fig:splines} symmetric polynomial B-splines for the first $\alpha$ degrees.\\

\begin{figure}
    \centering
    \includegraphics[width=0.4\textwidth]{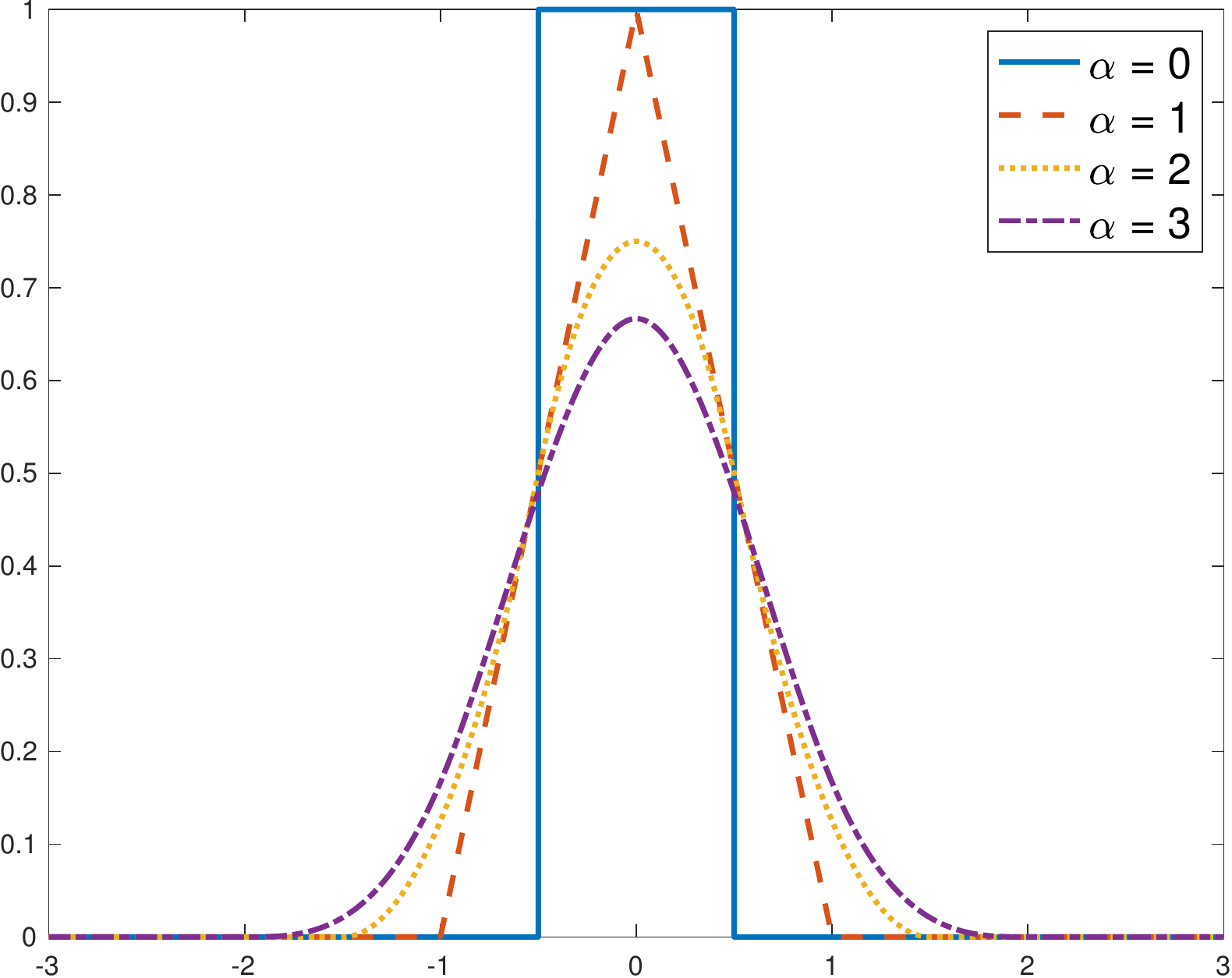}
    \caption{Symmetric polynomial B-splines of degree $\alpha = 0, \ldots , 3$.}
    \label{fig:splines}
\end{figure}

For cardinal splines, we use grid points $t_k = k - \frac{\alpha+1}{2}$ for $k \in \mathbb{Z}$. For an odd $\alpha$, we then have an integer grid, while an even $\alpha$ gives a half-integer grid. Additionally, the jump amplitudes $a[k]$ in \eqref{eq:L_spline} for $s = \beta^{\alpha}$ are denoted by $d_{\alpha}[k]$, a FIR digital filter corresponding to the finite difference of order $(\alpha+1)$ with knot locations at $t_k$ for $k \in [0 \ldots \alpha+1]$. We provide in Table \ref{tab:splines} a summary of the relevant characteristics of B-splines for small values of $\alpha$.

\begin{table*}[!t]
\centering
\caption{Characteristics of $\alpha^{\text{th}}$-degree B-splines.}
\label{tab:splines}
\begin{tabular}{l|l|l}
\hline
\hline
\multirow{4}{*}
\textrm{L}
& $\beta^{\alpha}(t)$ & 
\begin{math}
    d_{\alpha}[k], \ k = 0, \ldots, \alpha+1
\end{math}  \\
\hline
\begin{math}
\mathrm{D}^1
\end{math} &
$
\begin{aligned}
\beta^{0}(t) =
\begin{cases} 
1, & 0 \leq |t| < 1/2 \\
0, &\text{otherwise}
\end{cases} 
\end{aligned}$
& $(1, -1)$ \\

\begin{math}
\mathrm{D}^2 
\end{math} & 
$
\begin{aligned}
\beta^{1}(t) =
\begin{cases} 
1 - |t|, & 0 \leq |t| < 1\\
0, &\text{otherwise}
\end{cases}
\end{aligned}$
& 
$(1, -2, 1)$ \\

\begin{math}
\mathrm{D}^3
\end{math}
& 
$
\begin{aligned}
\beta^{2}(t)=
\begin{cases} 
3/4 - t^2, & |t| < 1/2\\
(3/2 - |t|)^2/2, & 1/2 \leq |t| < 3/2 \\
0, &\text{otherwise}
\end{cases}
\end{aligned}$
& $(1, -3, 3, -1) $\\

\begin{math}
\mathrm{D}^4 
\end{math}
& 
$
\begin{aligned}
\beta^{3}(t)=
\begin{cases} 
2/3 - |t|^2 + |t|^3/2, & |t| < 1\\
(2 - |t|)^3/6, & 1 \leq |t| < 2 \\
0, &\text{otherwise}
\end{cases}
\end{aligned}$
& $(1, -4, 6, -4, 1)$
\\
\hline
\hline
\end{tabular}
\end{table*}

\subsection{Discrete Formulation}

As suggested by Theorem \ref{Thm:Rep}, we take the stance of recasting the continuous-domain problem in \eqref{eq:continuous_problem} as a finite-dimensional optimization problem by restricting the search space to periodic L-splines with knots on a uniform grid. This allows us to effectively reduce the complexity of our algorithmic resolution. To do so, we describe our closed curves $\textbf{r}(t)$ as linear combinations of $N$ shifts of the $M$-periodized basis function $\varphi_{M} (t) = \sum_{k \in \mathbb{Z}} \varphi(\frac{t - Mk}{h})$, where $h = \frac{M}{N}$ is the grid stepsize. Following Section \ref{sec:b_splines}, we choose $\varphi_{M}$ to be the $M$-periodization and $h$-dilation of the B-spline generator $\varphi = \beta^{\alpha}$. These basis functions are weighted by two vectors of coefficients $\textbf{c}_x = (\mathrm{c}_x[n])^{N-1}_{n=0}$ and $\textbf{c}_y = (\mathrm{c}_y[n])^{N-1}_{n=0}$. Finally, the weighted functions are shifted by multiples of the grid size $h$ in order to describe

\begin{equation}\label{eq:curve_form}
    \textbf{r}(t) = \left[
    \begin{array}{cc}
        x(t)\\
        y(t)
    \end{array}
    \right]
    = 
    \left[
    \begin{array}{cc}
        \sum_{n = 0}^{N-1} \mathrm{c}_x[n] \ \varphi_{M}(t - nh) \\
        \sum_{n = 0}^{N-1} \mathrm{c}_y[n] \ \varphi_{M}(t - nh)
    \end{array}
\right].
\end{equation}

\subsection{Discrete Implementation}

Our choice \eqref{eq:curve_form} of curve parametrization allows us to optimize solely on the coefficients $\textbf{c}_x, \textbf{c}_y \in \R^{N}$ of two curve components. We implement a system matrix that interpolates the coefficients using $\varphi_{M}$ and samples them at the measurement locations. We introduce $\textbf{H} \in \R^{M \times N}$ with

\begin{equation} \label{eq:system_matrix}
    \left[\textbf{H}\right]_{m, n} = \varphi_{M}\left (m - nh \right).
\end{equation}

The regularization operator $\mathrm{L}$ becomes a circulant regularization matrix $\textbf{L} \in \R^{N \times N}$ composed of shifted versions of the $\alpha$th-order derivative operator coefficients $d_{\alpha}$ (see Table \ref{tab:splines}). The regularization matrix $\textbf{L}$ is therefore constructed as

\begin{equation}\label{eq:reg_matrix}
    \left[\textbf{L}\right]_{m, n} = \frac{1}{h^{\alpha}} d_{\alpha}[\left(m-n\right)_{\bmod N}].
\end{equation}

Our mixed-norm regularization involves, in the discrete setting, an $\ell_1-\ell_2$ norm given by

\begin{equation}
    \norm{[\textbf{f}_1 \ \ \textbf{f}_2]}_{\ell_1-\ell_2} \overset{\Delta}{=} \sum_{n = 0}^{N-1} \sqrt{(\textbf{f}_1[n])^2 + (\textbf{f}_2[n])^2},
\end{equation}

\noindent for $\textbf{f}_1, \textbf{f}_2 \in \R^N$. Indeed, we have that (see Theorem \ref{Thm:Intuition} in Appendix \ref{sec:appendix_norm_equivalence}): 

\begin{equation}\label{eq:norm_equivalence}
\norm{[\mathrm{L}\{x\} \ \mathrm{L}\{y\}]}_{\mathrm{TV}-\ell_2} = 
\norm{\textbf{L} \left[ \textbf{c}_x \ \textbf{c}_y \right] }_{\ell_1-\ell_2}. 
\end{equation}

Our discrete optimization problem therefore aims at finding $\textbf{c}_x$ and $\textbf{c}_y$ such that

\begin{equation}\label{eq:simple_min_problem}
    \argmin_{\textbf{c}_x, \textbf{c}_y \in \R^N} \norm{
        \begin{bmatrix}
            \textbf{H} & \textbf{0} \\
            \textbf{0} & \textbf{H}
        \end{bmatrix}
        \left[
        \begin{array}{cc}
            \textbf{c}_x \\
            \textbf{c}_y
        \end{array}
        \right]
        - 
        \left[
        \begin{array}{cc}
            \textbf{p}_x \\
            \textbf{p}_y
        \end{array}
        \right]
    }_2^2
    + 
    \lambda \norm{
        \textbf{L} \left[ \textbf{c}_x \ \textbf{c}_y \right]
    }_{\ell_1-\ell_2}.
\end{equation}

\subsection{Algorithmic Resolution}

To solve Problem \eqref{eq:simple_min_problem}, we use the alternating-direction method of multipliers (ADMM) solver \cite{BoydADMM} as implemented in the GlobalBioIm Matlab library \cite{GlobalBioIm} dedicated to the solution of inverse problems.

\section{Extension to Hybrid B-Spline Contours} \label{sec:hybrid_splines}

\subsection{Motivation and Continuous Model}

In Section \ref{sec:tv_based_optim}, we have presented a model and its implementation that were allowing the creation of a sparse curve using splines as basis functions. However, not all curves can be faithfully represented with a single type of spline. We propose to cater to this by modeling our closed function as a sum of two components $\textbf{r}(t) = \textbf{r}_1(t) + \textbf{r}_2(t)$. Similarly to the non-hybrid setting, we have $M$ points $\textbf{p}[m] = (\mathrm{p}_{x}[m], \ \mathrm{p}_{y}[m]), \ m = 0, \dots, M-1$. Hence, we again have that $\textbf{r}$ is $M$-periodic with $t \in \mathbb{T}_M$. Following the formulation for one-dimensional signals in \cite{DebarreAU19} and extending it to two dimensions, we consider continuous problems of the form

\begin{eqnarray}\label{eq:continuous_hybrid_problem}
    \nonumber \mathcal{V}_{\rm hyb}= && \argmin_{\substack{\textbf{r}_i \in  \mathcal{X}_{\mathrm{L}_i}(\mathbb{T}_M)  \\ \textbf{r}_1(0) = \textbf{0}}} 
    \sum_{m=0}^{M-1} \norm{\left.\textbf{r}_1(t)\right\vert_{t=m} + \left.\textbf{r}_2(t)\right\vert_{t=m} - \textbf{p}[m]}_2^2 \\
    && \mbox{}+ \lambda_1  \norm{\mathrm{L}_1\{\textbf{r}_1\}}_{\mathrm{TV}-\ell_2} + \lambda_2 \norm{\mathrm{L}_2\{\textbf{r}_2\}}_{\mathrm{TV}-\ell_2},
\end{eqnarray}

\noindent where $\lambda_1, \ \lambda_2 > 0$ are the two regularization parameters weighting of two regularization terms, and $\mathrm{L}_1$ and $\mathrm{L}_2$ are derivative operators of independent order. We now prove a representer theorem that suggests a parametric form for the optimal solution of Problem \ref{eq:continuous_hybrid_problem}. Its proof can be found in Appendix \ref{sec:app_thmhyb}.  \\

  \begin{theorem}\label{Thm:Hybrid}
  There exists a global minimizer ${\bf r}^*$ of \eqref{eq:continuous_hybrid_problem} that can be decomposed as ${\bf r}^* = {\bf r}^*_1 + {\bf r}^*_2$, where ${\bf r}^*_i$ are periodic ${\rm L}_i$-splines (see \eqref{eq:L_spline}) with $K_i$ knots, with $i=1,2$. Moreover, we have the bound $K_1+K_2\leq 2M$ for the  total number of knots of ${\bf r}^*$. 
\end{theorem}

The constraint $\textbf{r}_1(0) = \textbf{0}$ is necessary to handle the ill-posedness of the problem. Indeed, without this constraint, for any solution $(\textbf{r}_1(t), \textbf{r}_2(t))$ of Problem~\eqref{eq:continuous_hybrid_problem}, the pair $(\textbf{r}_1 + \textbf{v} , \textbf{r}_2 - \textbf{v})$, where $\textbf{v}$ is an arbitrary constant vector, would clearly also be a solution. This implies that the solution set would be unbounded, which can be problematic for numerical implementations. The constraint $\textbf{r}_1(0) = \textbf{0}$ resolves this issue without any restriction on the reconstructed curve, since any constant offset can be included in the $\textbf{r}_2$ component. See \cite{DebarreAU19} for more details concerning this question. 

\subsection{Discretization and Implementation}

As in Section \ref{sec:discrete_implementation}, we derive a discrete setting by using two sets of B-spline basis functions matched to their corresponding regularization operators. Given a grid of stepsize $h$, we consider closed $M$-periodic curves $\textbf{r}(t) = \textbf{r}_1(t) + \textbf{r}_2(t)$ such that, for $i = 1, 2$, we have

\begin{equation}
    \textbf{r}_i(t) = \left[
    \begin{array}{cc}
        x_{i}(t) \\
        y_{i}(t)
    \end{array}
    \right]
    = 
    \left[
    \begin{array}{cc}
        \sum_{n = 0}^{N-1} \mathrm{c}^i_{x}[n] \ \varphi^{i}_{M}(t - nh) \\
        \sum_{n = 0}^{N-1} \mathrm{c}^i_{y}[n] \ \varphi^{i}_{M}(t - nh)
    \end{array}
\right].
\end{equation}

The two regularization operators are set to $\mathrm{L}_i = \mathrm{D}^{\alpha_{i}+1}$, with $\mathrm{D}^{\alpha_{i}+1}$ the derivative operator of order $\alpha_{i}+1$ for $i = 1, 2$. As in Section \ref{sec:discrete_implementation}, these operators lead to spline solutions \cite{DebarreAU19}. We hence set $\varphi^{i}_{M}$ to B-splines of degrees $\alpha_{i}$ for $i = 1, 2$ and with $\alpha_1 < \alpha_2$. \\

This choice of curve allows us to optimize over the coefficients $\textbf{c}^1_{x} = (\mathrm{c}^1_x[n])^{N-1}_{n=0}$, $\textbf{c}^1_{y} = (\mathrm{c}^1_y[n])^{N-1}_{n=0}$, $\textbf{c}^2_{x} = (\mathrm{c}^2_x[n])^{N-1}_{n=0}$, and $\textbf{c}^2_{y} = (\mathrm{c}^2_y[n])^{N-1}_{n=0}$. 
As in Section \ref{sec:tv_based_optim}, we can define $\textbf{H}_{i}$ the system matrices producing and sampling the continuous curves $\textbf{r}_i$, as well as their corresponding regularization matrices $\textbf{L}_{i}$. We construct $\textbf{H}_{i}$ and $\textbf{L}_{i}$ as in \eqref{eq:system_matrix} and \eqref{eq:reg_matrix}, respectively. Finally, we cater to the constraint $\textbf{r}_1(0) = \textbf{0}$ by enforcing that $(\textbf{c}^1_x * \textbf{b}_{\alpha_1})[0] = 0$ and $(\textbf{c}^1_y * \textbf{b}_{\alpha_1})[0] = 0$, where $\mathrm{b}_{\alpha_1}[k] = \varphi^1_M(kh)$ and $*$ denotes a cyclic discrete convolution.  \\

Our discrete hybrid-optimization problem takes the form

\begin{eqnarray}
&& \argmin_{\substack{\textbf{c}^1_{x}, \textbf{c}^1_{y}, \\ \textbf{c}^2_{x}, \textbf{c}^2_{y} \in \R^N}} \quad \norm{
        \begin{bmatrix}
            \textbf{H}_1 & \textbf{0} & \textbf{H}_2 & \textbf{0} \\
            \textbf{0} & \textbf{H}_1 & \textbf{0} & \textbf{H}_2 
        \end{bmatrix}
        \left[
        \begin{array}{cc}
            \textbf{c}^1_{x} \\
            \textbf{c}^1_{y} \\
            \textbf{c}^2_{x} \\
            \textbf{c}^2_{y} 
        \end{array}
        \right]
        - 
        \left[
        \begin{array}{cc}
            \textbf{p}_x \\
            \textbf{p}_y
        \end{array}
        \right]
    }_2^2 \nonumber \\
&&\mbox{}+ \lambda_1 \norm{\textbf{L}_1 \left[ \textbf{c}^1_{x} \ \textbf{c}^1_{y}\right]}_{\ell_1-\ell_2} + \lambda_2 \norm{ \textbf{L}_2 \left[ \textbf{c}^2_{x} \ \textbf{c}^2_{y}\right]}_{\ell_1-\ell_2}, \label{eq:min_hybrid_problem}
\end{eqnarray}

\noindent subject to

\begin{align}
    & \mbox{} \nonumber (\textbf{c}^1_x * \textbf{b}_{\alpha_1})[0] = 0,\\
& \mbox{} (\textbf{c}^1_y * \textbf{b}_{\alpha_1})[0] = 0.
\end{align}

As in Section \ref{sec:discrete_implementation}, we use the ADMM solver \cite{BoydADMM} to find a solution to Problem \eqref{eq:min_hybrid_problem} and GlobalBioIm \cite{GlobalBioIm} to implement our algorithms.

\section{Experiments}\label{sec:experiments}

We evaluate the distance between the constructed curves and the contour points through the quadratic fitting error (QFE) defines as

\begin{equation}\label{eq:QFE}
    \text{QFE} = \frac{1}{M} \sum_{m = 0}^{M-1} \norm{\left.\textbf{r}(t)\right\vert_{t=m} - \textbf{p}[m]}_2^2.
\end{equation}

It is noteworthy that the QFE can be used at the same time in the single-spline setting and in the hybrid setting. Indeed, by replacing the hybrid curve $\textbf{r} = \textbf{r}_1 + \textbf{r}_2$ in \eqref{eq:QFE}, we obtain a QFE that is consistent with the data-fidelity term in \eqref{eq:continuous_hybrid_problem}.

 For computational reasons, we chose the lowest resolution, i.e. the largest grid size $h$, that allowed us to solve the problem in a satisfactory way, thus effectively making $h$ a hyperparameter. In this work, the number of knots $N$ was chosen so that it matched the order of magnitude of the number of data points. It is important to note that increasing $N$, thus splitting the grid, can only improve the solution in terms of cost. 

\subsection{Rotation Invariance}

To verify that our regularization norm is truly rotation-invariant, we apply a planar rotation of angle $\theta$ to our data before we reconstruct the curve with the regularization operator $\mathrm{L} = \mathrm{D}^2$. We have added to the data a Gaussian perturbation with a signal-to-noise ratio (SNR) of $47.28$ dB. We compare the curves reconstructed with our regularization to the curves resulting from the $\left(\text{TV-}\ell_1\right)$ regularization of Definition \ref{Def:Regularization}. Indeed, $\ell_1$ regularization is widely used in the signal-processing community as a sparsifying prior. To do so, we choose $\lambda$ in the non-rotated $\left(\text{TV-}\ell_1\right)$ regularized curve (Figure \ref{l1_M_rot0}) so that the QFE matches the QFE from the non-rotated RI-TV regularized curve (Figure \ref{mixed_M_rot0}). When rotating the measurements, we adjust $\lambda$ again so that the QFE of the $\left(\text{TV-}\ell_1\right)$-regularized curve on the rotated points matches the one of the curve constructed with RI-TV regularization with rotated points. We see in Figure \ref{fig:rot_invariance} that the RI-TV-regularized problem provides the same solution regardless of $\theta$. Indeed, the knot locations do not differ between Figures \ref{mixed_M_rot0} and \ref{mixed_M_rot40}, nor does the number $K$ of knots. This is not the case for the purely $\left(\text{TV-}\ell_1\right)$-regularized problem. Not only are the knot locations different when a rotation is applied to the measurements, but the number $K$ of knots varies with $\theta$ as well as the QFE of the curve. Additionally, one needs to adapt $\lambda$ to obtain the same QFE between the constructed curves on rotated and non-rotated measurements.

\begin{figure*}[t!]
    \centering
    
    \subfloat[RI-TV regularization, $\theta = 0^{\circ}$, $K = 20$,
    $\lambda = 700$, 
    $\text{QFE} = 12.09$.]{\includegraphics[width=0.50\textwidth]{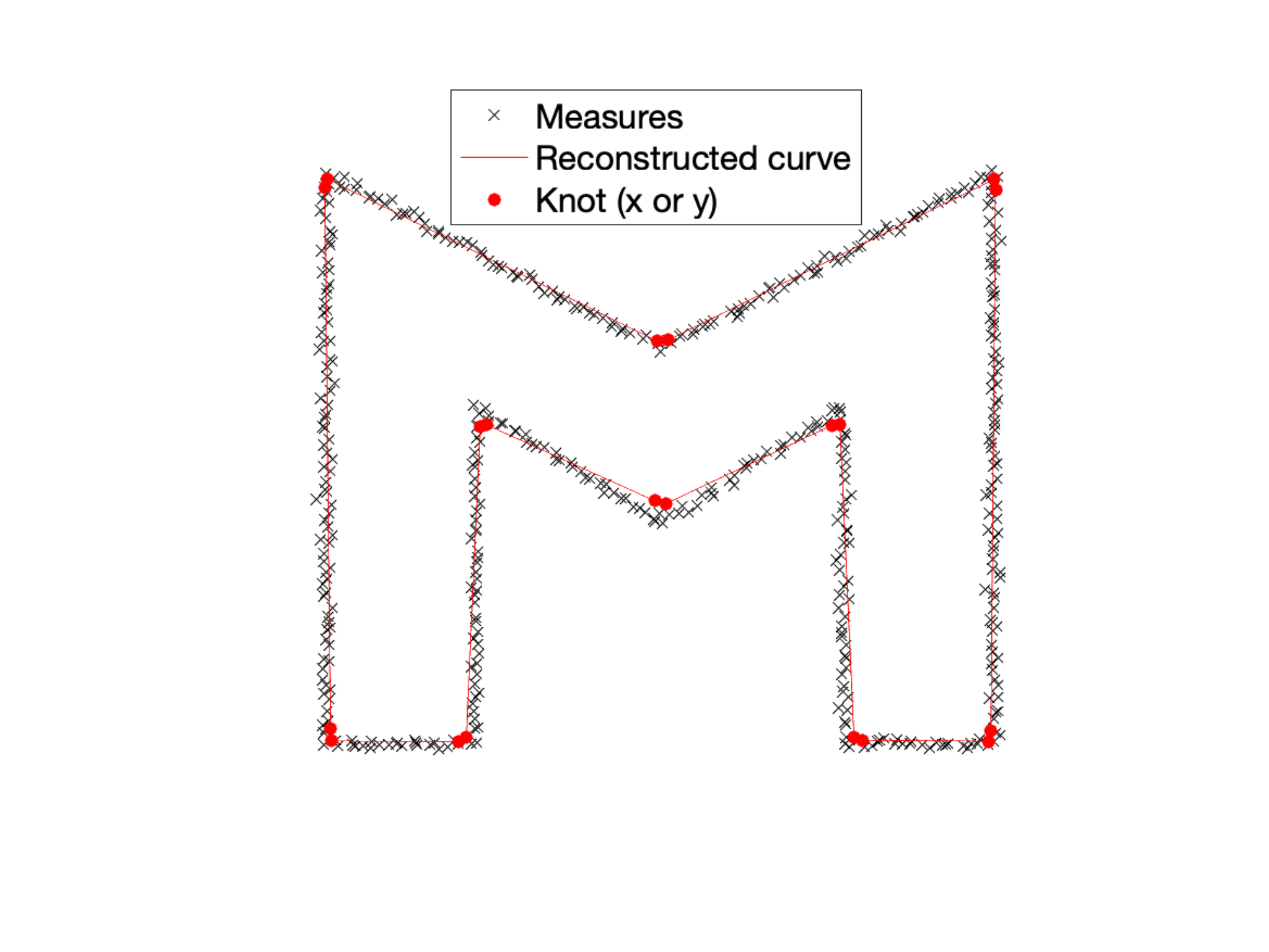}\label{mixed_M_rot0}}
    \hfill
    \subfloat[RI-TV regularization, $\theta = 40^{\circ}$, $K = 20$, 
    $\lambda = 700$, 
    $\text{QFE} = 12.09$.]{\includegraphics[width=0.50\textwidth]{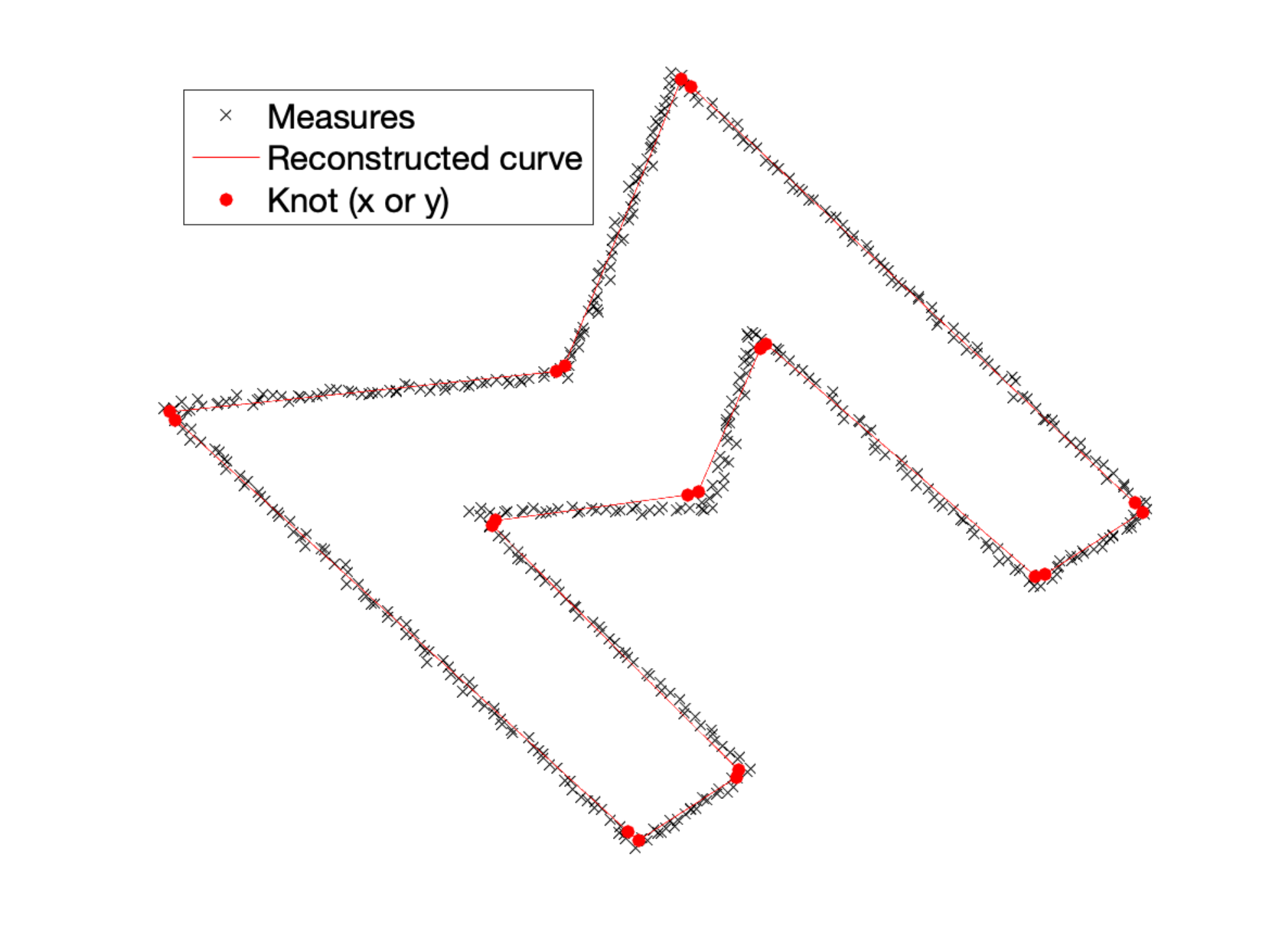}\label{mixed_M_rot40}}
    
    \subfloat[$\left(\text{TV-}\ell_1\right)$ regularization, $\theta = 0^{\circ}$, $K = 37$, 
    $\lambda = 482.13$, 
    $\text{QFE} = 12.09$.]{\includegraphics[width=0.50\textwidth]{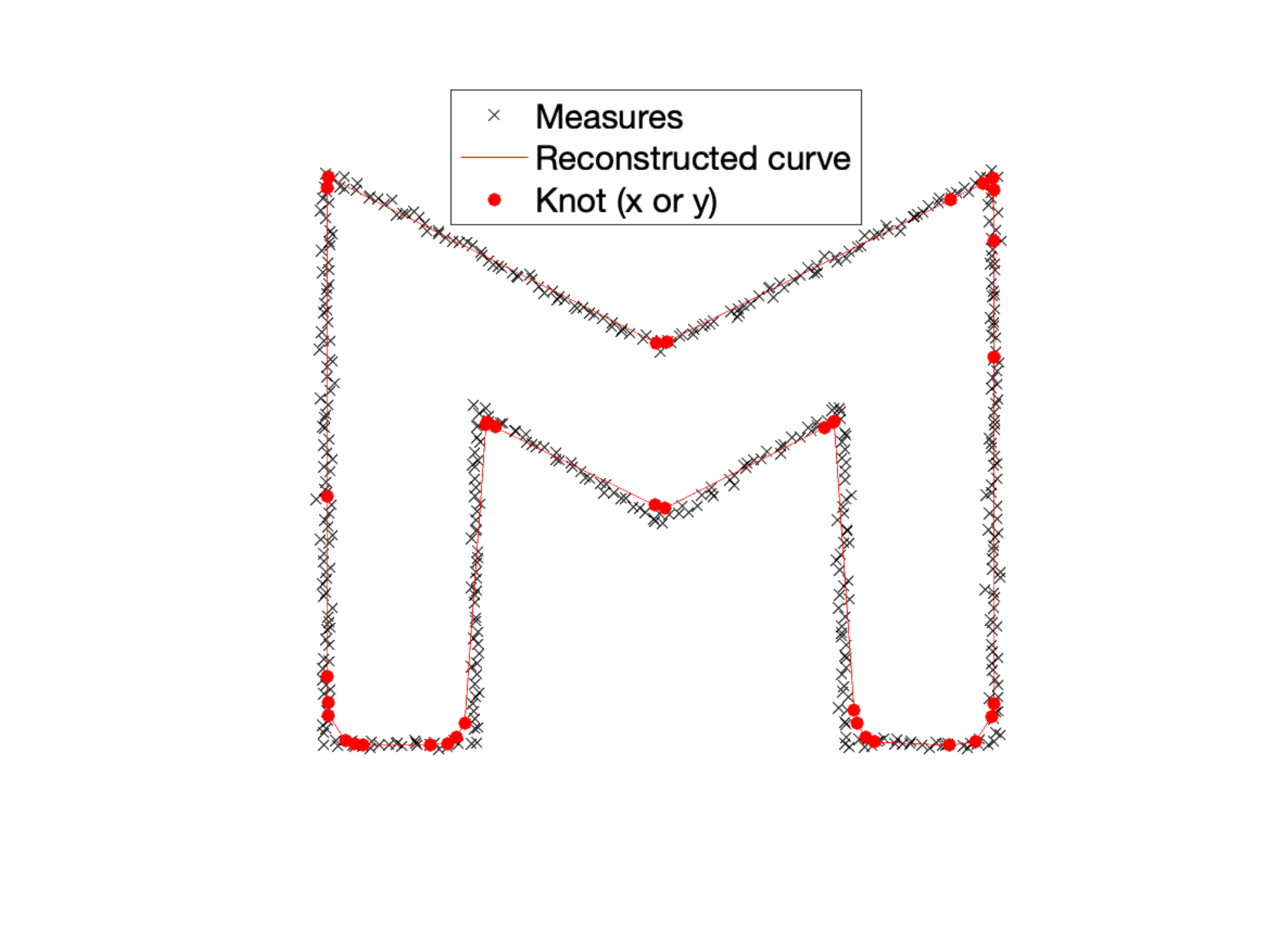}\label{l1_M_rot0}}
    \hfill
    \subfloat[$\left(\text{TV-}\ell_1\right)$ regularization, $\theta = 40^{\circ}$, $K = 29$, 
    $\lambda = 500.93$, 
    $\text{QFE} = 12.09$.]{\includegraphics[width=0.50\textwidth]{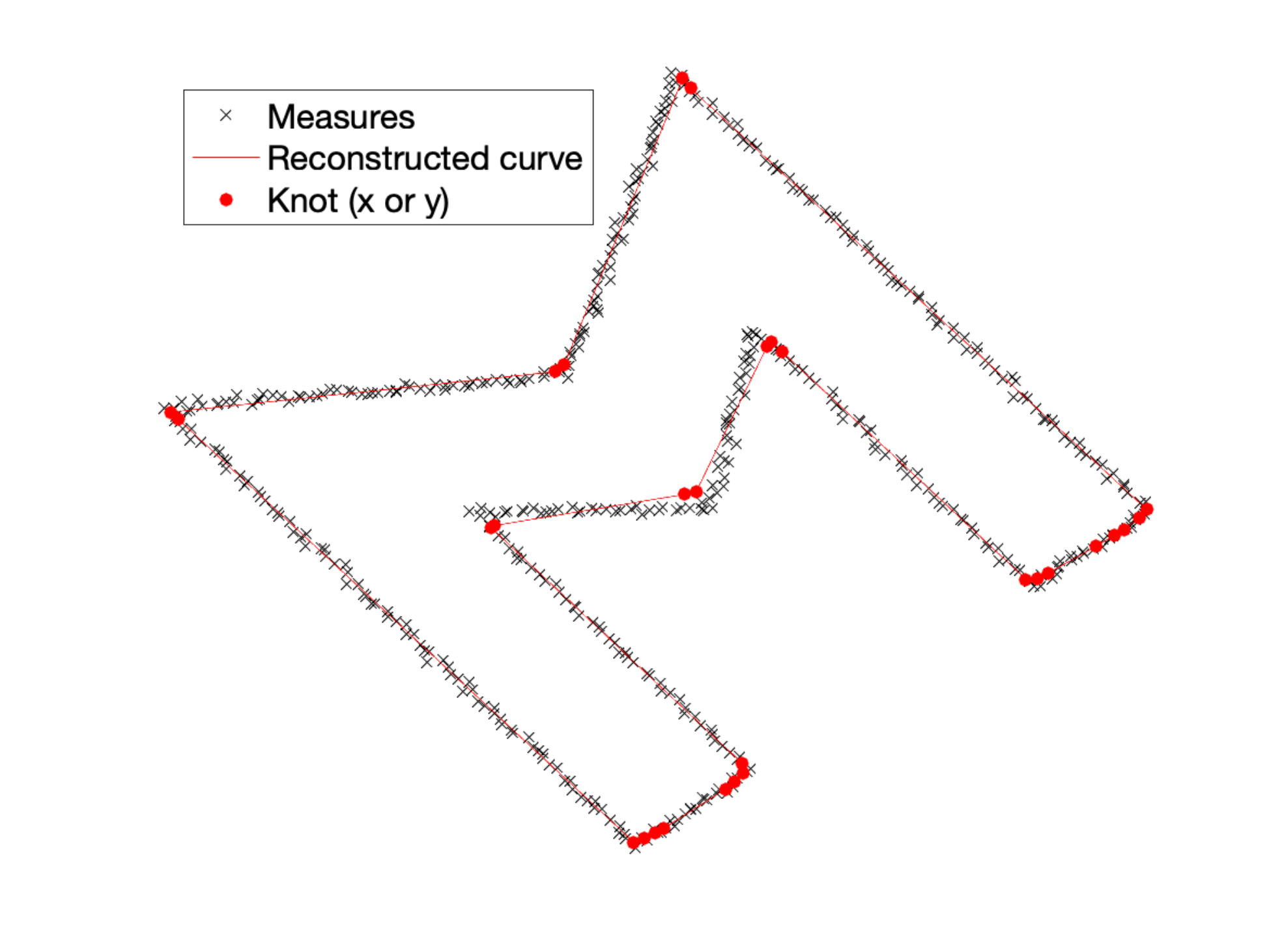}\label{l1_M_rot40}}
    
    \caption{Solutions as a function of the rotation angle $\theta$ for RI-TV regularization and $\left(\text{TV-}\ell_1\right)$ regularization for a same contour. $M = 488$, grid stepsize $h = 1.9062$, $\varphi = \beta^1$.}
    \label{fig:rot_invariance}
    
\end{figure*}

\subsection{Sparsity}

We compared the results obtained using our framework to a widespread technique for curve approximation, which is knot-removal algorithms for B-spline curves. As a reference, we chose the method implemented in the NURBS-Python (geomdl) library \cite{bingol2019geomdl}, which allows the user to input data points and the number of knots desired and outputs a B-spline curve. The data fed to both approaches has additive Gaussian noise with a signal-to-noise ratio (SNR) of $47.28$ dB. For a fair comparison, we equalized the number of knots in both solutions. The results are shown in Figure \ref{fig:sparsity}. We observe that the curve obtained with our solution is very close to the data points, and that the knots are placed at relevant locations where the underlying signal's singularities are expected. The solution provided by NURBS-Python (geomdl) when the number of knots is set to the same as our solution, albeit being rotation-invariant, is far from the data points. The knot locations do not seem to follow the underlying signal's expected singularities. This is confirmed by the difference in QFE. Indeed, the curve obtained with the RI-TV-regularized problem has a QFE of $\text{QFE} = 12.09$ relative to the data points, versus a QFE of $\text{QFE} = 153.47$ for the NURBS-Python curve. This result shows that our framework is able to maximize the sparsity of the constructed curve while remaining faithful to the data. 

\begin{figure}
    \centering
    
    \subfloat[RI-TV regularization, $K = 20$,
    $\lambda = 700$, 
    $\text{QFE} = 12.09$.]{\includegraphics[width=0.50\textwidth]{figures/mixed_M_rot0-eps-converted-to.pdf}\label{mixed_M_rot0_}}
    \hfill
    
    \subfloat[NURBS-Python (geomdl) curve approximation, $K = 20$, 
    $\text{QFE} = 153.47$.]{\includegraphics[width=0.50\textwidth]{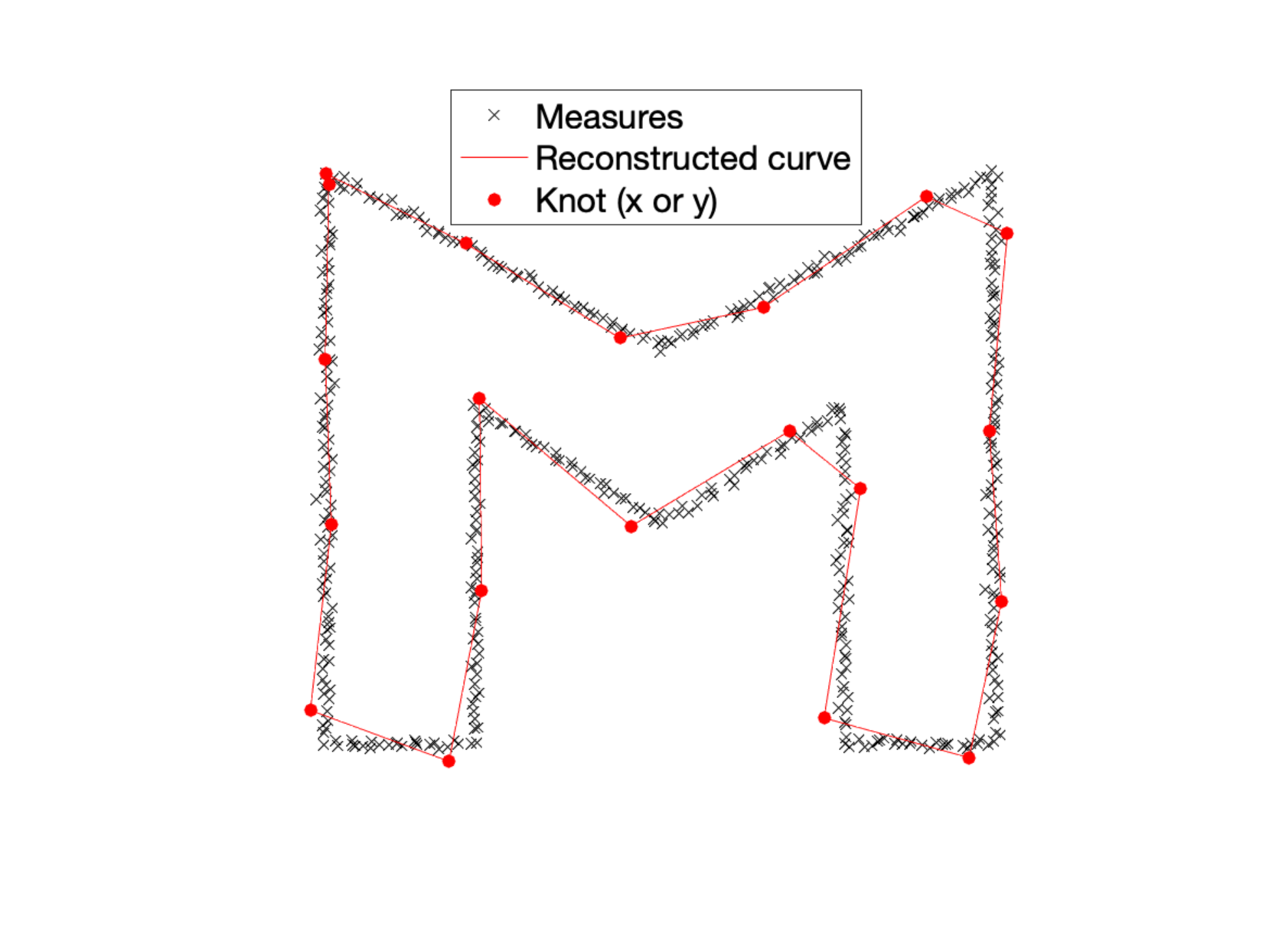}\label{geomdl_M_rot0}}
    
    \caption{Solutions for our framework and NURBS-Python (geomdl) for a same contour and with the same number of knots. $M = 488$, $\varphi = \beta^1$.}
    \label{fig:sparsity}
    
\end{figure}

\subsection{Resilience to Contour Imprecisions}

A beneficial feature derived from the enforcement of joint sparsity in the two curve components is resilience of our reconstructions to imprecisions in the contour points. Indeed, when we expect our data to be imprecise, we can choose to increase the regularization parameter $\lambda$ at the cost of data fidelity, as the curve cannot rely as much on the data. Particularly, when the regularizer is TV-based, an increase in $\lambda$ tends to smoothen sharp variations. This is visible in Figure \ref{fig:noise_resilience}, where several curves have been reconstructed using linear B-splines $\beta^1$. 
Figures \ref{fig:mixed_M_nonoise}, \ref{fig:mixed_M_sigma1}, and \ref{fig:mixed_M_sigma2} are reconstructions of increasingly inaccurate measurements using RI-TV. Figures \ref{fig:l1_M_nonoise}, \ref{fig:l1_M_sigma1}, and \ref{fig:l1_M_sigma15} depict reconstructions resulting from a sparsifying regularization without coupling $\left(\text{TV-}\ell_1\right)$, using a $\lambda$ tuned so that the QFE matches the QFE of the curves regularized by RI-TV. When TV regularization is used, and as the contour becomes more inaccurate, the number $K$ of knots drastically increases and the angles are deformed. On the contrary, for the reconstructions in Figures \ref{fig:mixed_M_nonoise}, \ref{fig:mixed_M_sigma1}, even as the noise and $\lambda$ increase, the number $K$ of knots remains unchanged and the angles are fairly well preserved.

\begin{figure*}[t!]
    \centering
    
    \subfloat[RI-TV regularization, no noise, $\lambda = 8$, $K = 20$, $\text{QFE} = 5.86$.]{\includegraphics[width=0.32\textwidth]{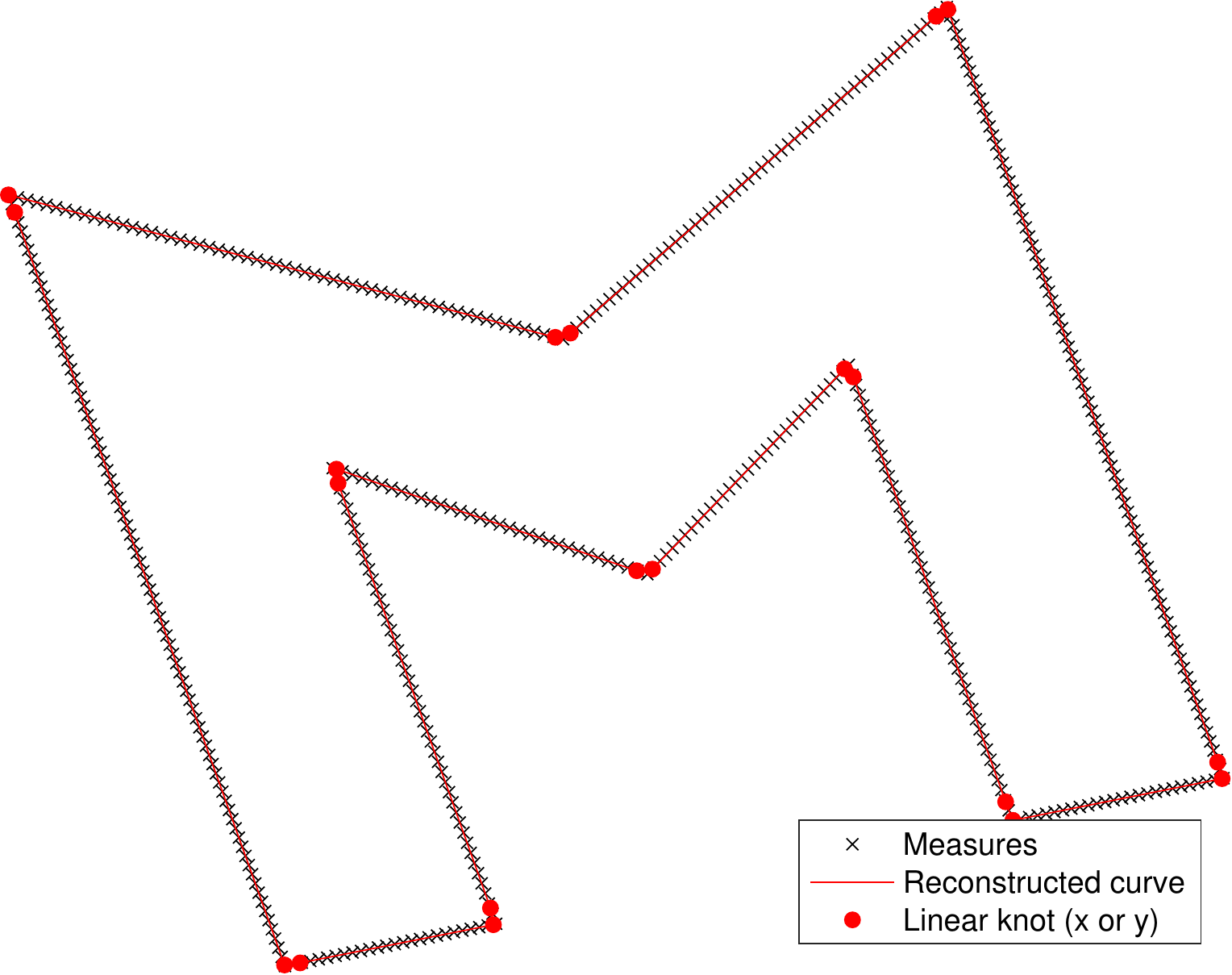}\label{fig:mixed_M_nonoise}}
    \hfill
    \subfloat[RI-TV regularization, SNR $ = 47.05$ dB, $\lambda = 700$, $K = 20$, $\text{QFE} = 12.14$.]{\includegraphics[width=0.32\textwidth]{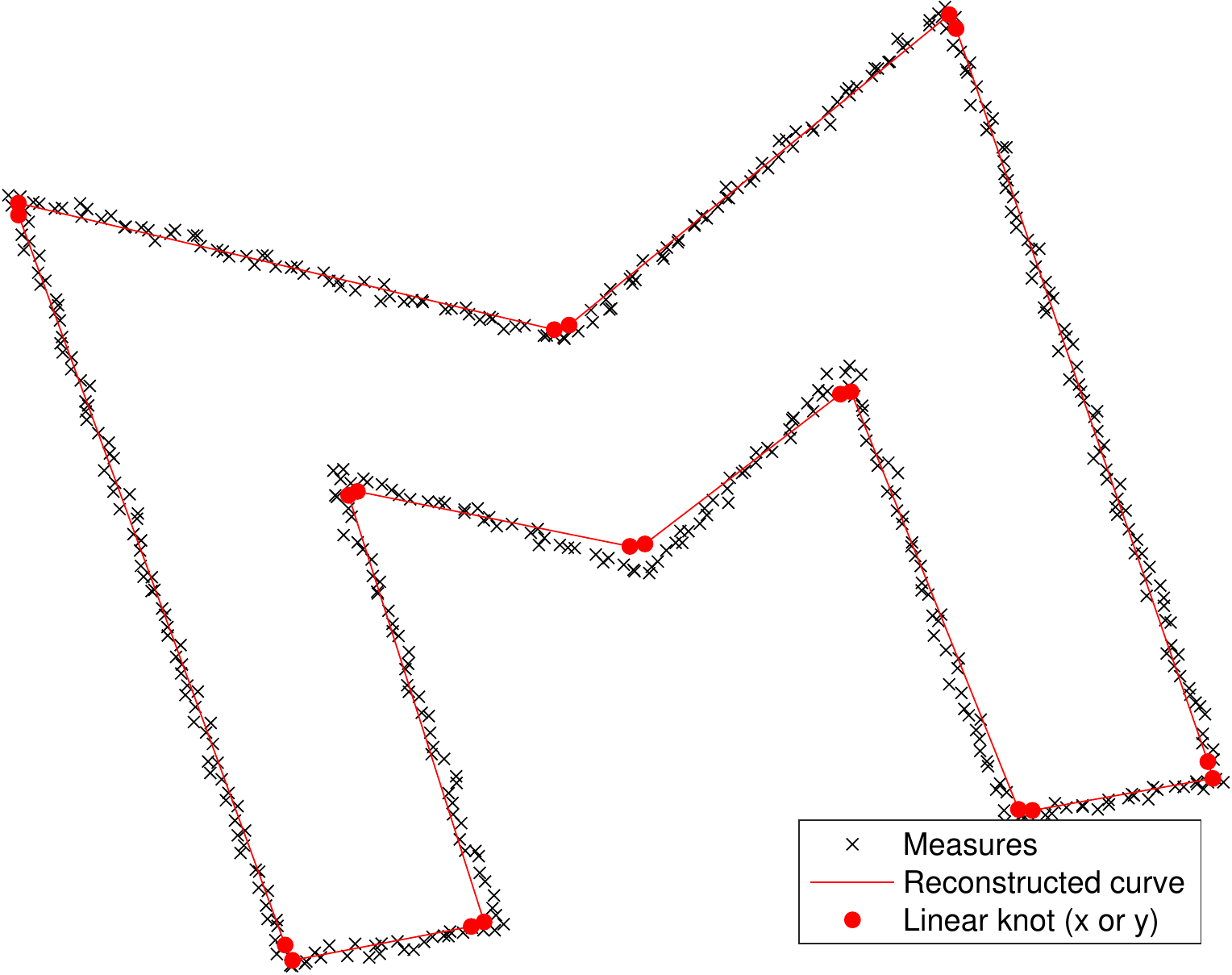}\label{fig:mixed_M_sigma1}}
    \hfill
    \subfloat[RI-TV regularization, SNR $ = 41.20$ dB, $\lambda = 800$, $K = 20$, $\text{QFE} = 18.95$.]{\includegraphics[width=0.32\textwidth]{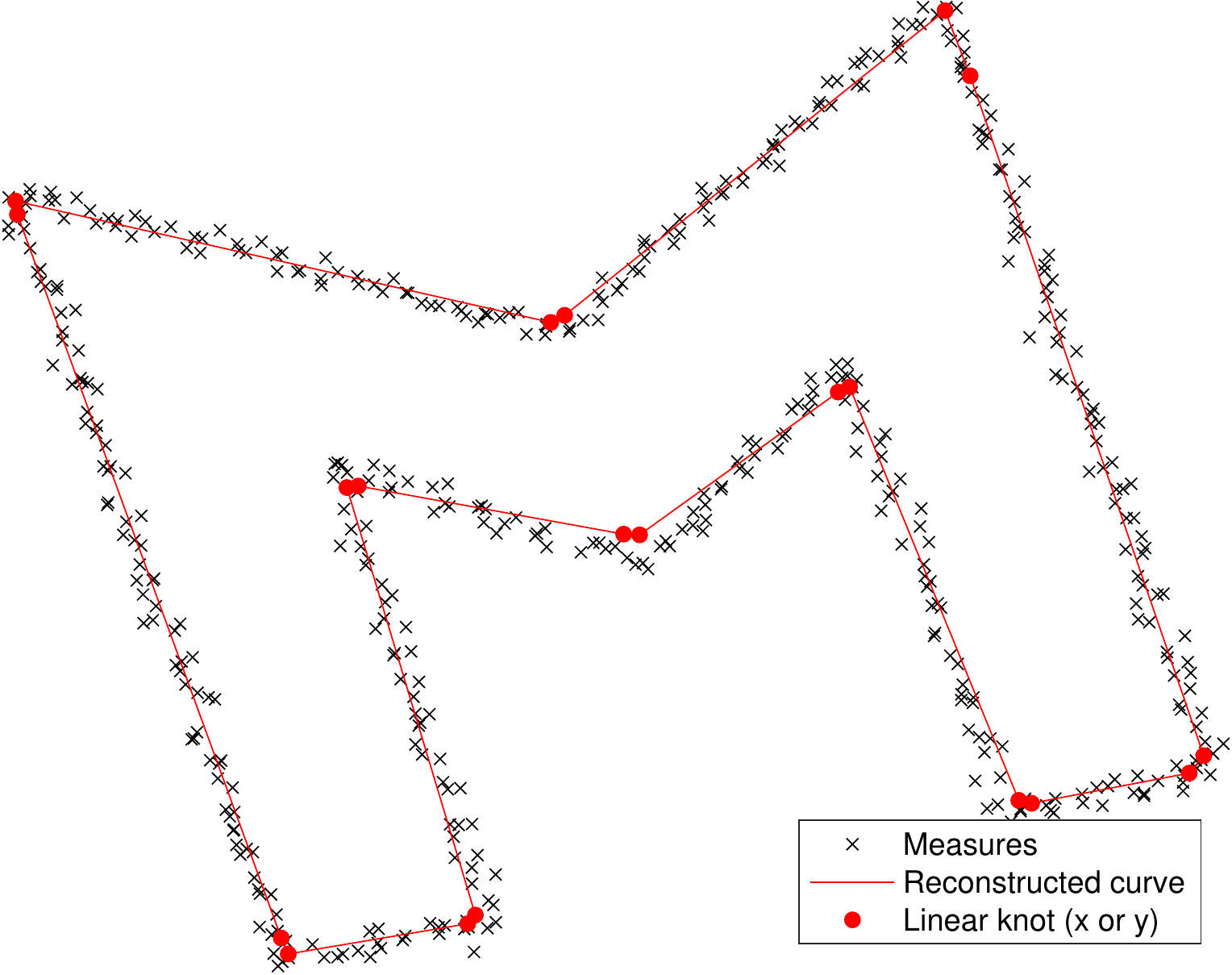}\label{fig:mixed_M_sigma2}}
    
    \subfloat[$\left(\text{TV-}\ell_1\right)$ regularization, no noise, $\lambda = 10$, 
    $K = 20 $, $\text{QFE} = 5.86$.]{\includegraphics[width=0.32\textwidth]{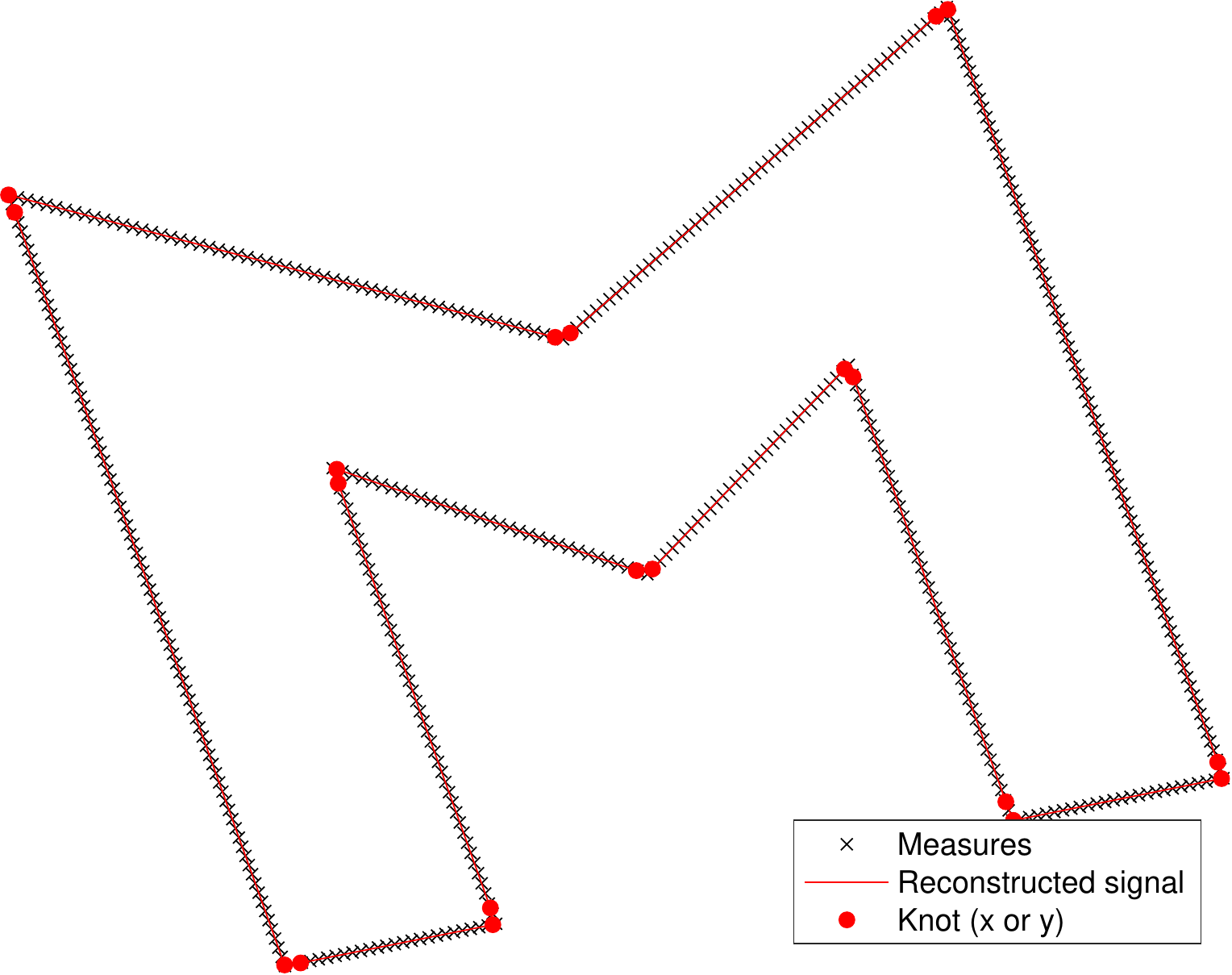}\label{fig:l1_M_nonoise}}
    \hfill
    \subfloat[$\left(\text{TV-}\ell_1\right)$ regularization, SNR $ = 47.05$ dB, 
    $\lambda = 459.45$, $K = 36$, $\text{QFE} = 12.14$.]{\includegraphics[width=0.32\textwidth]{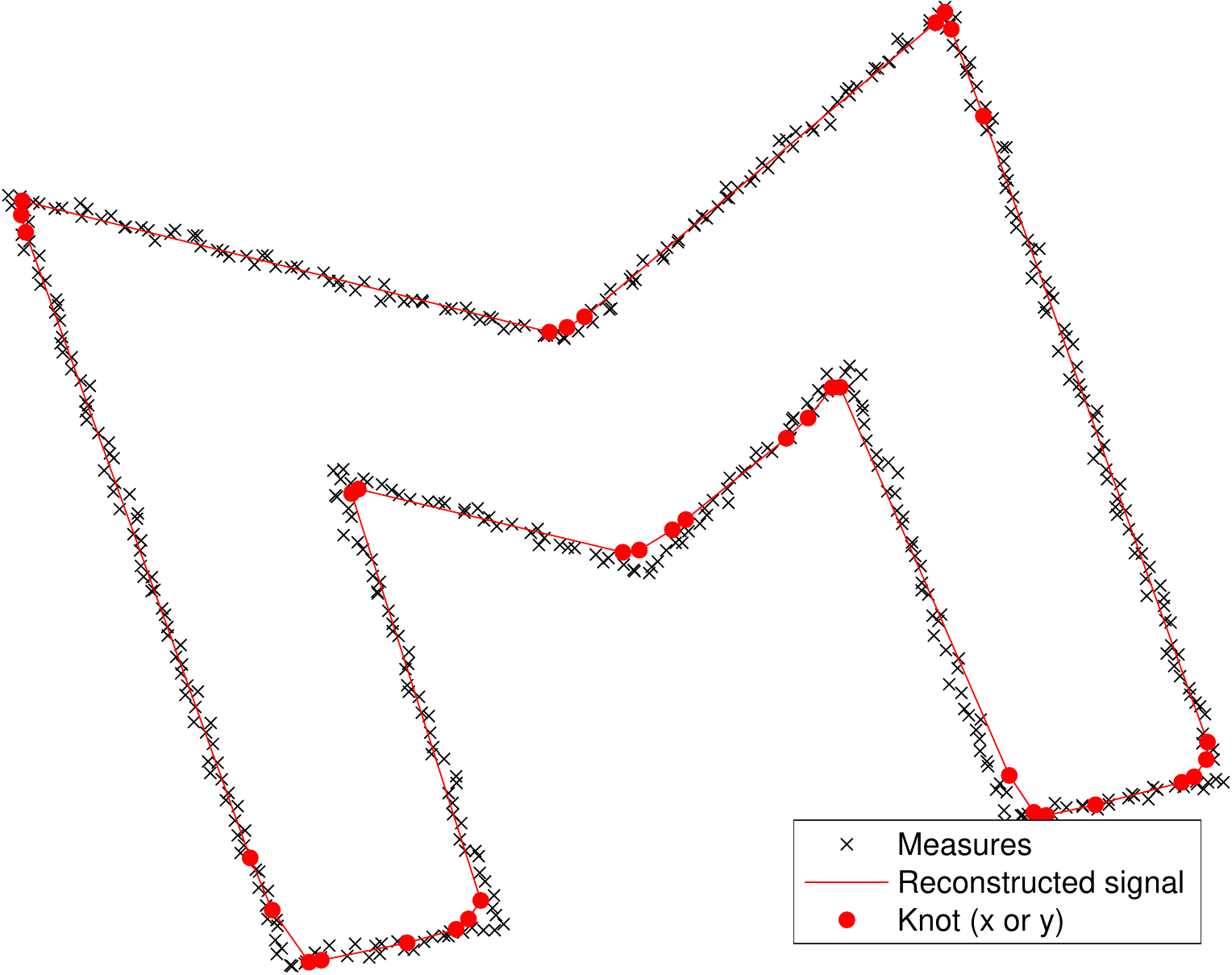}\label{fig:l1_M_sigma1}}
    \hfill
    \subfloat[$\left(\text{TV-}\ell_1\right)$ regularization, SNR $ = 41.20$ dB, 
    $\lambda = 531.35$, $K = 35$, $\text{QFE} = 18.95$.]{\includegraphics[width=0.32\textwidth]{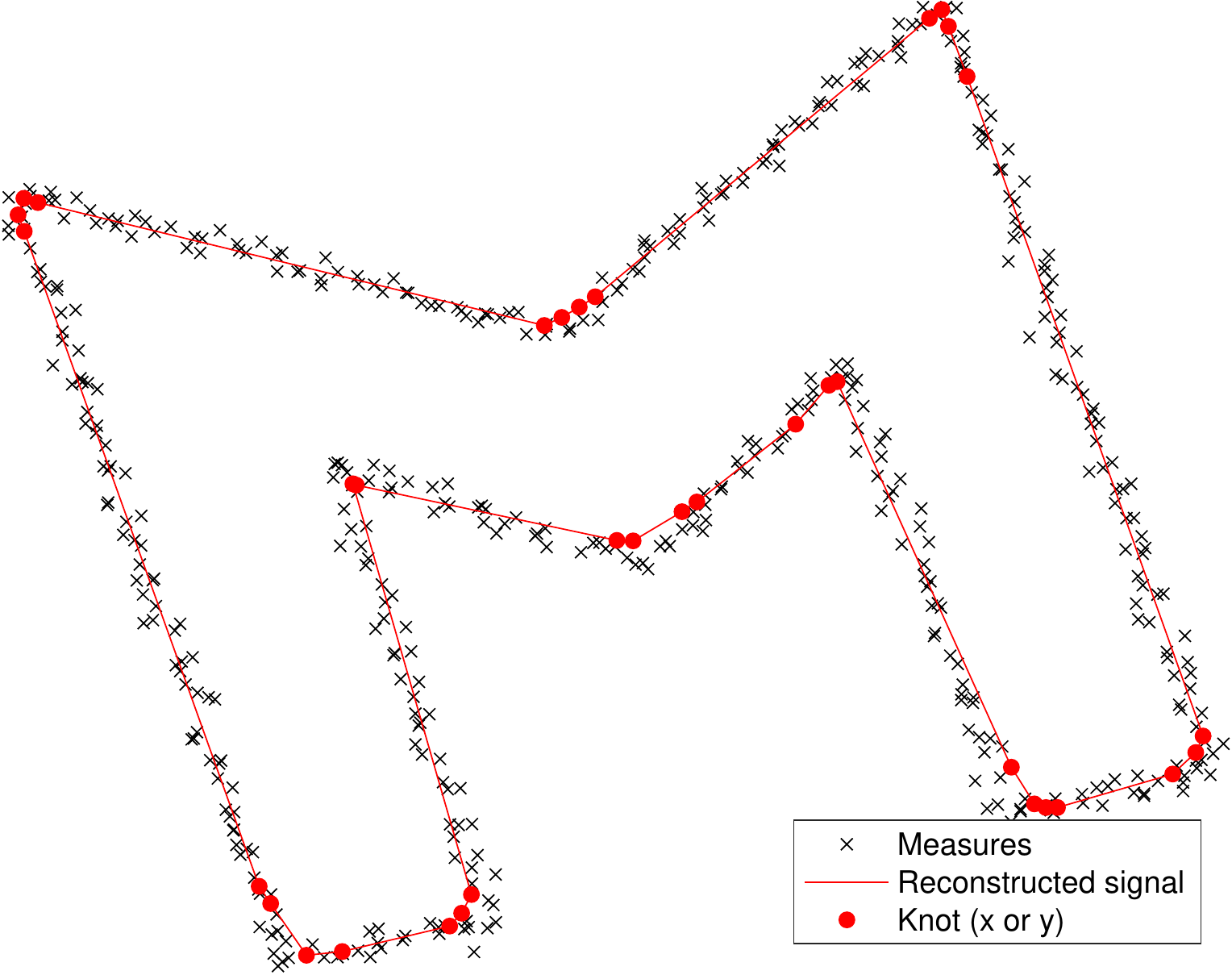}\label{fig:l1_M_sigma15}}
    
    \caption{Resilience to noise for RI-TV regularization and $\left(\text{TV-}\ell_1\right)$ regularization.}
    \label{fig:noise_resilience}
\end{figure*}

\subsection{Hybrid Setting Applications}

The single-component framework described in Sections \ref{sec:tv_based_optim} and \ref{sec:discrete_implementation} only allows for the use of one kind of B-spline per curve. However, when the contour under consideration is composed of smooth sections and kinks, no single type of B-spline can provide a faithful and sparse reconstruction. An example of curve fitting that depicts this problem is given in Figure \ref{fig:hybrid}. We reconstructed the contour using first $\beta^1$ and $\beta^3$ as basis functions, giving piecewise-linear and piecewise-cubic curves in Figures \ref{fig:mixed_G_linear} and \ref{fig:mixed_G_cubic}, respectively. Figure \ref{fig:mixed_G_hybrid} contains a reconstruction under the hybrid setting $\mathrm{L_1} = \mathrm{D}^2$ and $\mathrm{L_2} = \mathrm{D}^4$, thus producing a curve that has both a linear and a cubic component. While all three reconstructions yield the same QFE with respect to the data, the hybrid curve has by far the smallest number of knots. Moreover, upon visual inspection, the hybrid curve in Figure \ref{fig:mixed_G_hybrid} portrays the most faithful reconstruction, as it does round neither the angles nor the straight lines, nor does it straighten the smooth sections. 

\begin{figure}
\centering
    \subfloat[Spline degree: $1$, $\lambda = 31.87$, $K = 89$.]{
    \includegraphics[width=0.48\textwidth]{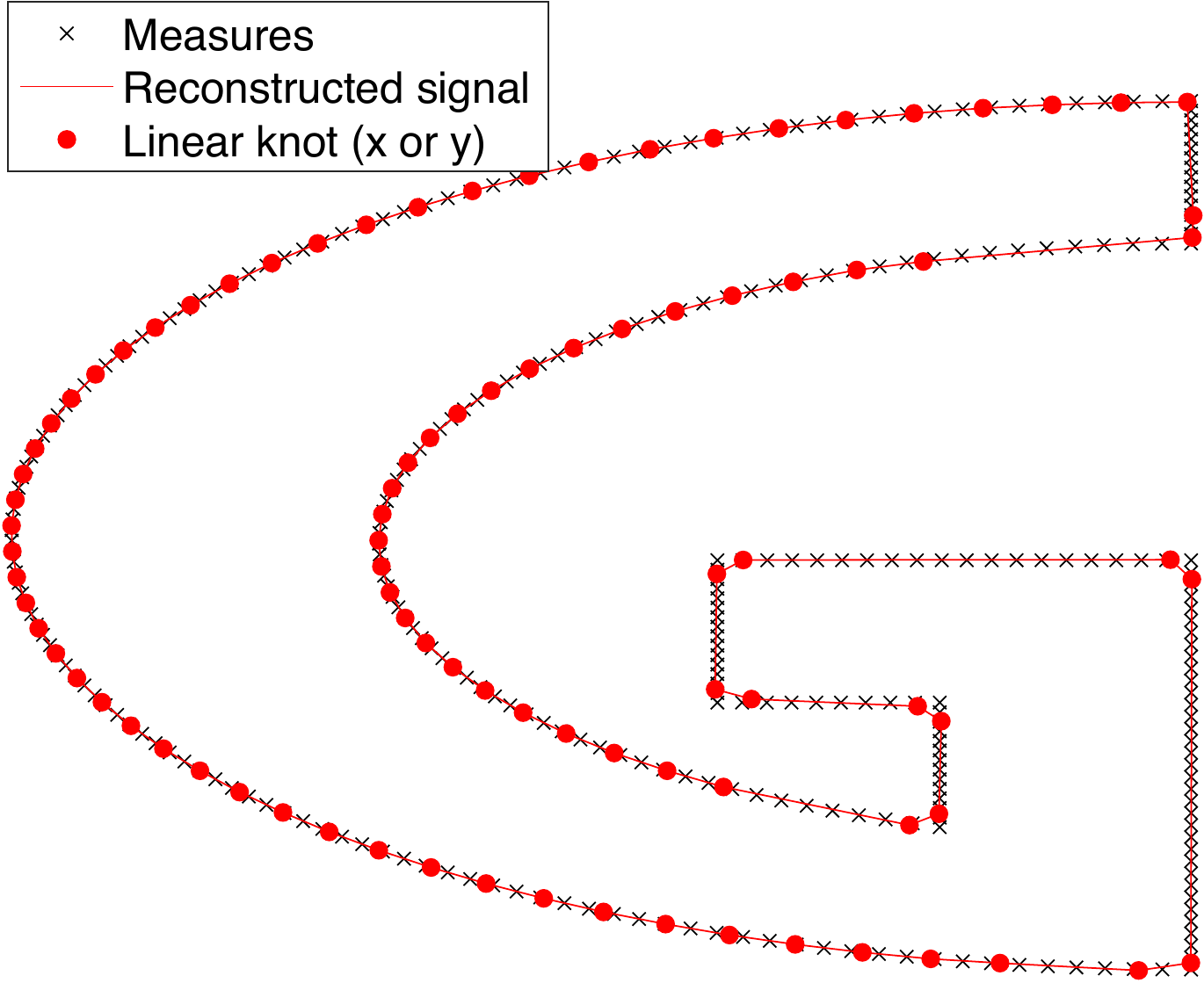}\label{fig:mixed_G_linear}}
    \hfill
    \subfloat[Spline degree: $3$, $\lambda = 24.72$, $K = 44$.]{\includegraphics[width=0.48\textwidth]{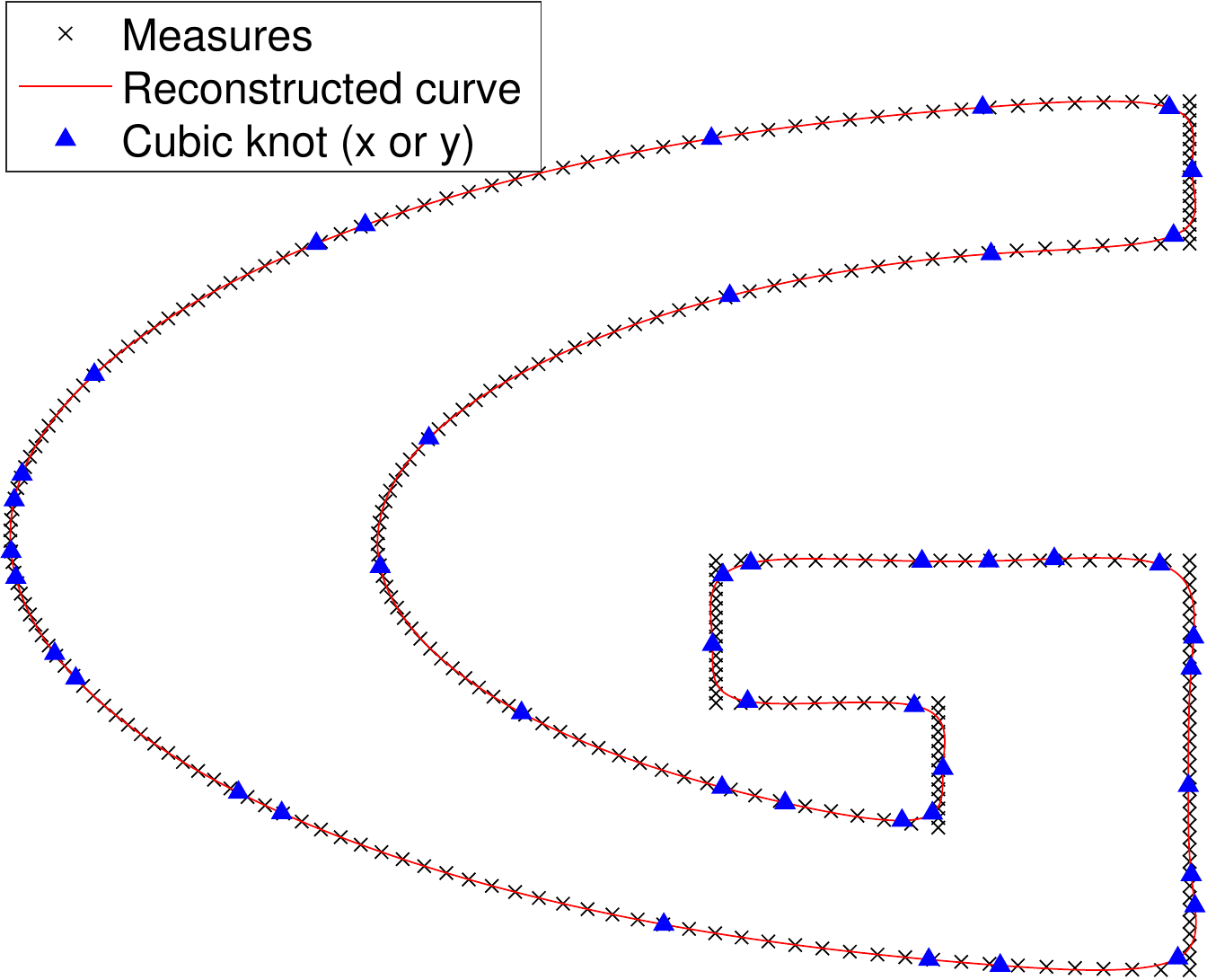}\label{fig:mixed_G_cubic}}
    \hfill
    \subfloat[Spline degrees: $1$ and $3$, $\lambda_1 = 80$, $\lambda_2 = 90$,  $K = 37$]{\includegraphics[width=0.48\textwidth]{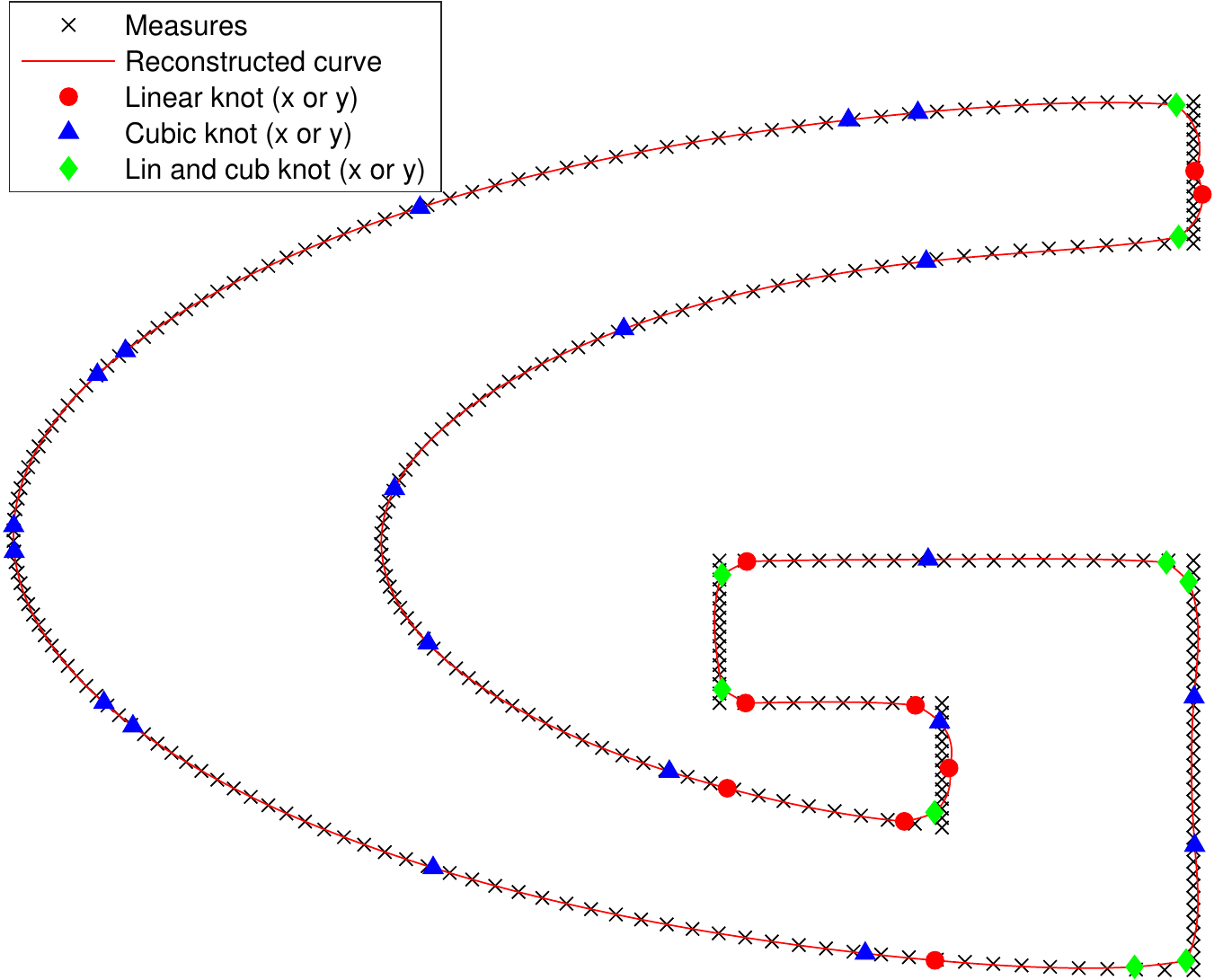}\label{fig:mixed_G_hybrid}}%
\caption{Noiseless curve reconstruction with a single spline, a hybrid setting, and RI-TV regularization. All three reconstructions have a constant $\text{QFE} = 8.88$.}
\label{fig:hybrid}
\end{figure}

We can observe the effect of the parameters $\lambda_1$ and $\lambda_2$ on the constructed curve when the hybrid reconstruction setting is applied to real contour points for a constant ratio of knots $\frac{K_1}{K_2} = 0.86$. 
In Figure \ref{fig:daftpunk}, as $\lambda_1$ and $\lambda_2$ increase, the total number $K$ of knots decreases and the curve becomes more stylized. In addition, for all values of $\lambda_1$ and $\lambda_2$, our algorithm preserves the kinks of the contour while mimicking its smooth segments.  

\begin{figure*}
    \centering
    
    \subfloat[Data. ]{\includegraphics[width=0.48\textwidth]{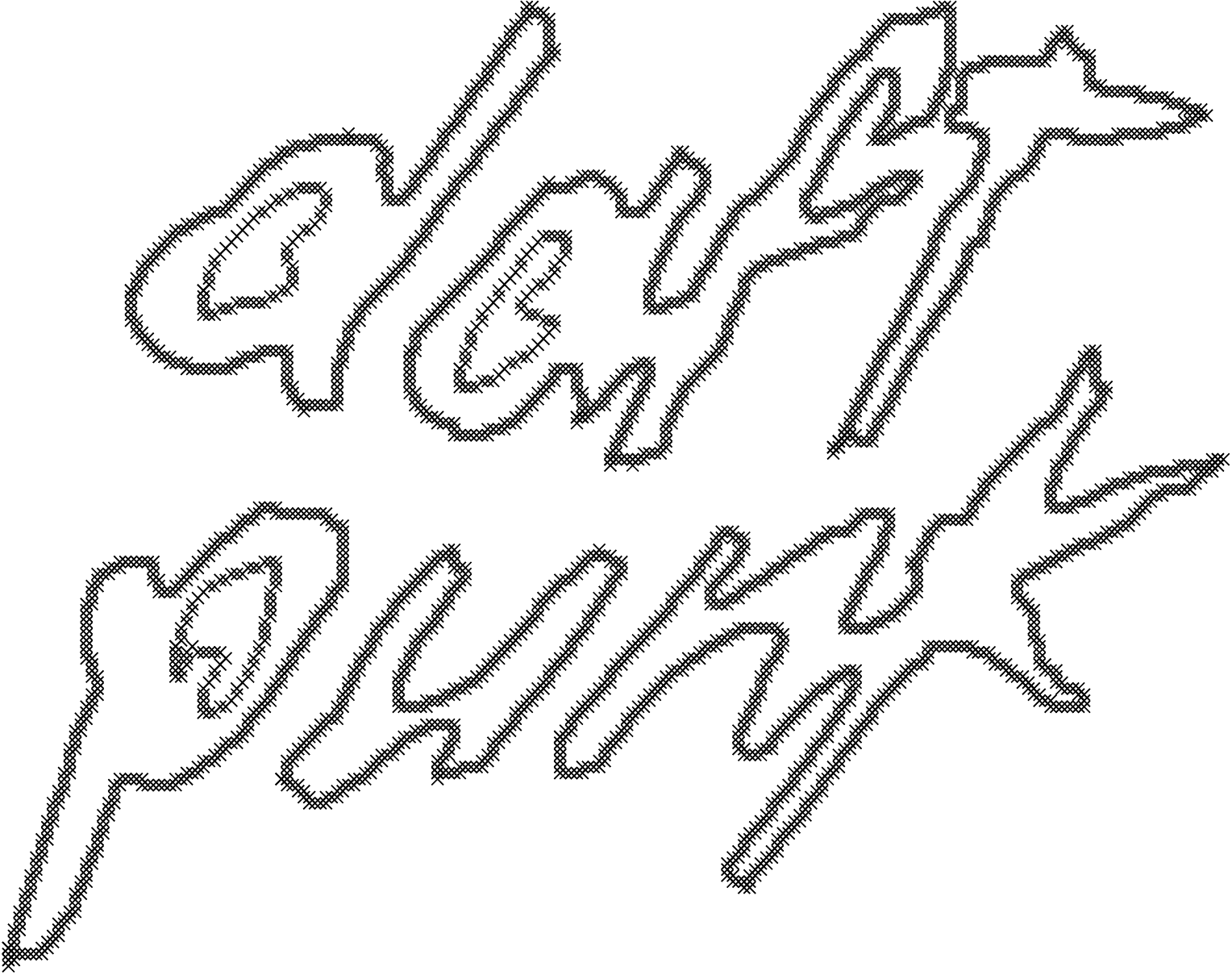}\label{fig:daftpunk_measures}}
    \hfill
    \subfloat[$\lambda_1 = 5$, $\lambda_2 = 95$, $K = 312$, $\text{QFE} = 0.80$.
    ]{\includegraphics[width=0.48\textwidth]{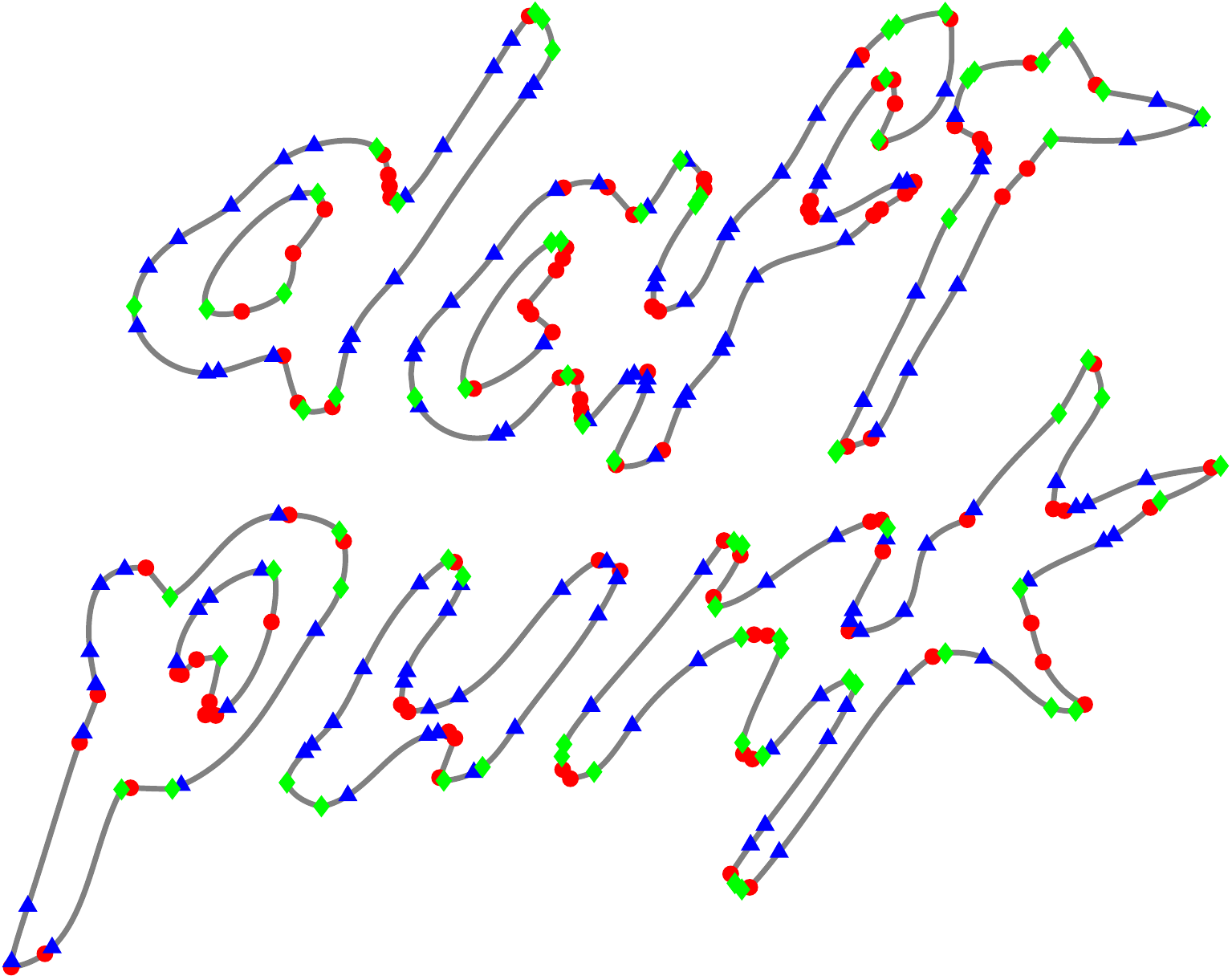}\label{fig:daftpunk_l100}}
    \hfill
    \subfloat[$\lambda_1 =20$, $\lambda_2 = 980$, $K = 229$, $\text{QFE} = 1.11$.
    ]{\includegraphics[width=0.48\textwidth]{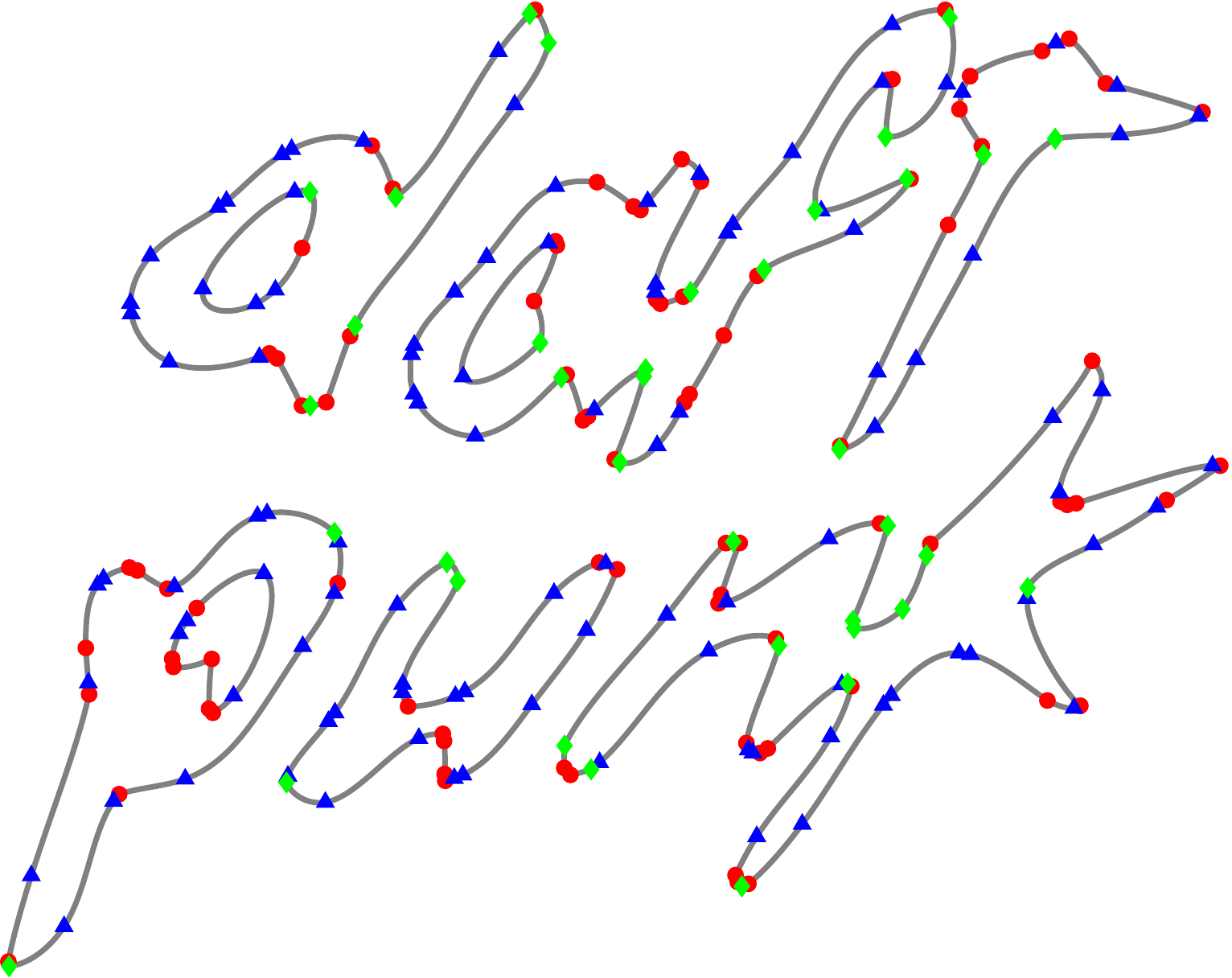}\label{fig:daftpunk_l1000}}
    \hfill
    \subfloat[$\lambda_1 =100$, $\lambda_2 = 9900$, $K = 139$, $\text{QFE} = 2.82$.]{\includegraphics[width=0.48\textwidth]{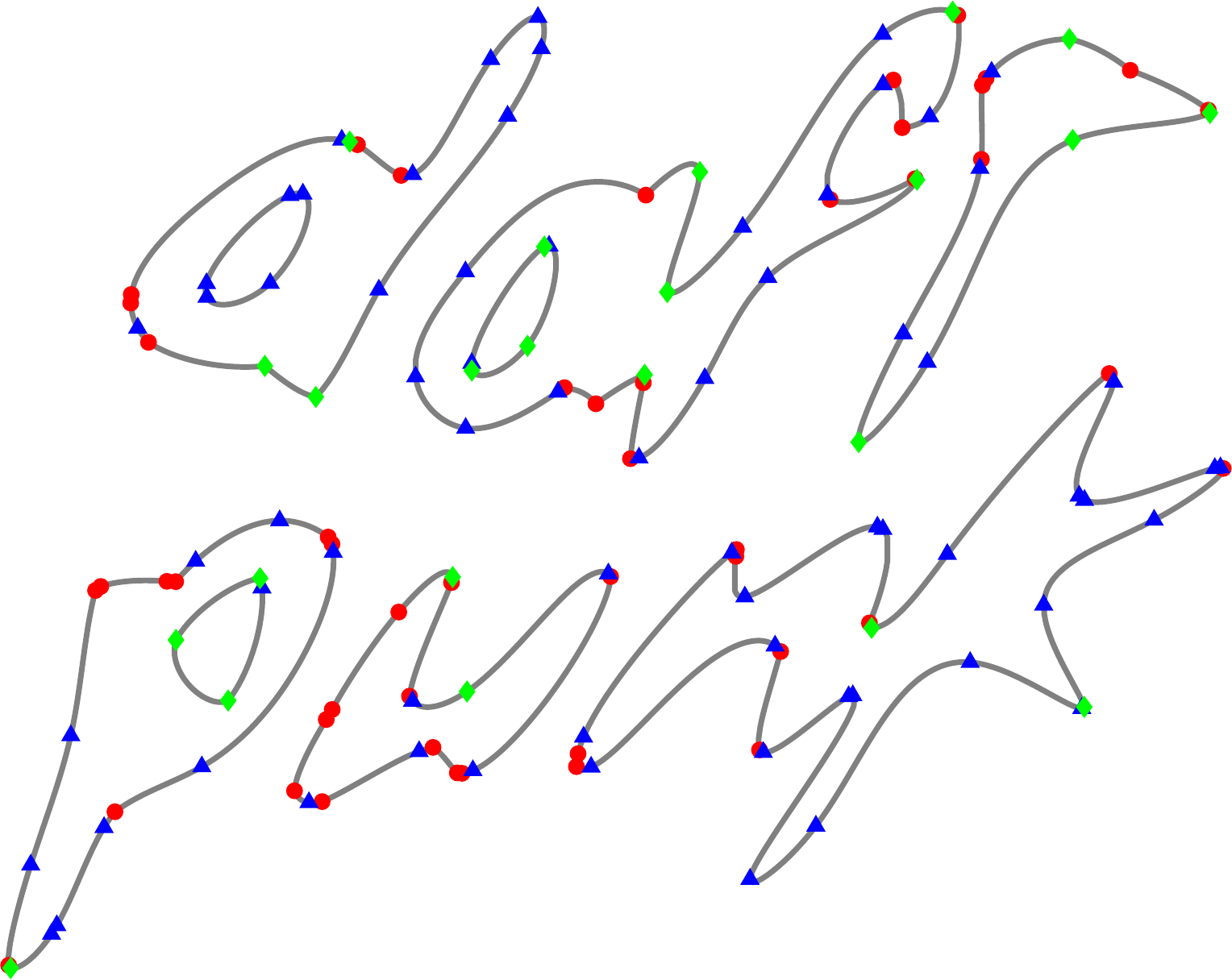}\label{fig:daftpunk_l10000}}
    
    \caption[DP wiki]{Effect of $\lambda_1$ and $\lambda_2$ on the reconstructed hybrid curve for $M = 2714$ under RI-TV regularization. The reconstructed curve is represented by the solid line. The round markers and the triangular markers indicate the location of the linear and cubic knots, respectively. The diamond-shaped markers indicate the superimposition of a linear knot and a cubic knot. The data are extracted from the official Daft Punk logo\footnotemark.}
    \label{fig:daftpunk}
\end{figure*}

\section{Conclusion}
We have introduced a framework to reconstruct sparse continuous curves from a list of possibly inaccurate contour points using an RI-TV regularization norm. We have proved that an optimal solution to our minimization problem is a curve that uses splines as basis functions, and we have leveraged this result to provide an exact discretization of the continuous-domain framework using B-splines. Furthermore, we have extended our formulation to reconstruct curves with components of distinct smoothness properties using hybrid splines. We have experimental confirmation of the rotation invariance of our regularizer. In addition, our experimental results demonstrate that our formulation yields sparse reconstructions that are close to the data points even when their noise increases, unlike other regularization methods. Finally, our hybrid-curve experiments demonstrate that we are able to faithfully reconstruct any contour with a low number of knots.

\appendices

\section{Proof of Theorem \ref{Thm:Intuition}}\label{sec:appendix_norm_equivalence}

\begin{proof}
{\bf Item} \ref{Item:L1}: Let $\boldsymbol{\varphi}=(\varphi_1,\varphi_2)\in \mathcal{S}(\mathbb{T}_M)^2$ be an arbitrary smooth curve with $\|\boldsymbol{\varphi}\|_{q,\infty}=1$. On the one hand,    the H\"older inequality for vectors implies that, for any $t\in \mathbb{T}_M$, 
\begin{equation}\label{Eq:CauchySchwartz}
 \left|f_1(t)\varphi_1(t) + f_2(t)\varphi_2(t) \right| \leq \|{\bf f}(t)\|_p \|\boldsymbol{\varphi}(t)\|_q\leq \|{\bf f}(t)\|_p,
\end{equation}
where the last inequality is due to  $\|\boldsymbol{\varphi}\|_{q,\infty}=1$. On the other hand, the inclusion $f_i \in L_1(\mathbb{T}_M)$ allows us to express the duality product $\langle f_i, \varphi_i \rangle$ as a simple integral of the form 
\begin{equation}\label{Eq:duality_product}
    \langle f_i, \varphi_i \rangle = \int_{0}^M f_i(t)\varphi_i(t){\rm d}t, \quad i=1,2.
\end{equation}

Combining \eqref{Eq:duality_product} with \eqref{Eq:CauchySchwartz}, we obtain that
\begin{align}
 \langle f_1,\varphi_1\rangle + \langle f_2, \varphi_2 \rangle &= \int_0^M \left(f_1(t)\varphi_1(t)+f_2(t)\varphi_2(t)\right){\rm d}t \nonumber \\&\leq  \int_0^M \|{\bf f}(t)\|_p {\rm d}t.
\end{align}
Taking the supremum over all $\boldsymbol{\varphi}\in \mathcal{S}(\mathbb{T}_M)^2$ with $\|\boldsymbol{\varphi}\|_{q,\infty}=1$ then yields that $
\norm{[f_1 \ \ f_2]}_{\mathrm{TV}-\ell_p} \leq  \int_0^M \|{\bf f}(t)\|_p{\rm d} t$. To prove the equality, we first define the functions 
\begin{equation}
g_i:\mathbb{T}_M\rightarrow\mathbb{R}: t \mapsto \mathbbm{1}_{{\bf f}(t)\neq \boldsymbol{0}} \frac{{\rm sgn}(f_i(t)) |f_i(t)|^{(p-1)}}{\|{\bf f}(t)\|_p^{(p-1)}}, 
\end{equation}
where $\mathbbm{1}_A$ denotes the indicator function of the set $A$. We note that $g_i,i=1,2$ are Borel-measurable  with $\|g_i\|_{L_\infty} \leq 1$ for $i=1,2$. Further, one readily verifies that 
\begin{equation}\label{Eq:Estim1}
 \int_0^M \left(f_1(t) g_{1}(t)+ f_2(t)g_2(t)\right) {\rm d}t   = \int_0^M \|{\bf f}(t)\|_p {\rm d}t. 
\end{equation}
By invoking a variant of Lusin's theorem (see \cite[Theorem 7.10]
{folland1999real}) on the space $\mathcal{C}(\mathbb{T}_M)$  of M-periodic continuous functions, we then consider the $\epsilon$-approximations $
{g}_{i,\epsilon} \in \mathcal{C}(\mathbb{T}_M)$ of $g_i$ such that $\|
{g}_{i,\epsilon}\|_{L_\infty} \leq \|g_i\|_{L_\infty}\leq 1$, and $
\int_{E}  |f_i(t)|{\rm d}t \leq \epsilon/8,  i=1,2$, 
where  $E=\{t\in \mathbb{T}_M: g_{i,\epsilon}(t) \neq g_i(t)\}$. This in effect implies that 

\begin{align}
&\int_0^M \left| f_i(t)\right| \cdot\left|g_{i,\epsilon}(t)-g_{i}(t)  \right|{\rm d}t \nonumber \\&\qquad= \int_{E} |f_i(t)|\cdot \left|g_{i,\epsilon}(t)-g_{i}(t)\right| {\rm d}t \nonumber \\& \qquad\leq \|\mathbbm{1}_E f_i\|_{L_1} \|g_{i,\epsilon} -g_i\|_{L_{\infty}} \leq \frac{\epsilon}{4} \label{Eq:Estim2}.
\end{align}

We then invoke the denseness of $\mathcal{S}(\mathbb{T}_M)$ in $\mathcal{C}(\mathbb{T}_M)$ to deduce the existence of $\varphi_{i,\epsilon} \in \mathcal{S}(\mathbb{T}_M)$ with $\|g_{i,\epsilon} -\varphi_{i,\epsilon} \|_{L_\infty} \leq \frac{\epsilon}{4\|f_i\|_{L_1}}$. This gives us the upper-bound 

\begin{equation}\label{Eq:Estim3}
\int_0^M \left|f_i(t)\right| \cdot \left|\varphi_{i,\epsilon}(t)-g_{i,\epsilon}(t)\right| {\rm d}t  \leq \|f_i\|_{L_1} \|\varphi_{i,\epsilon}-g_{i,\epsilon}\|_{L_\infty} \leq  \frac{\epsilon}{4}.
\end{equation}

Next, we use the triangle inequality to obtain the lower-bound

\begin{align}
\langle f_i , \varphi_{i,\epsilon} \rangle &\geq \int_0^M f_i(t) g_i(t){\rm d}t \nonumber \\ & \quad - \int_0^M |f_i(t)|\cdot |g_i(t)-g_{i,\epsilon}(t)|{\rm d}t \nonumber \\& \quad -  \int_0^M |f_i(t)|\cdot |g_{i,\epsilon}(t)-\varphi_{i,\epsilon}(t)|{\rm d}t \nonumber
\\& \geq  \int_0^M f_i(t)g_i(t){\rm d}t  - \frac{\epsilon}{2}, \quad i=1,2,
\end{align}

where the last inequality follows the combination of \eqref{Eq:Estim2} and \eqref{Eq:Estim3}. Finally, we use \eqref{Eq:Estim1} to conclude that 
\begin{equation}
    \norm{[f_1 \ \ f_2]}_{\mathrm{TV}-\ell_p} \geq \frac{\langle f_1,\varphi_{1,\epsilon} \rangle + \langle f_2,\varphi_{2,\epsilon} \rangle }{\|(\varphi_{1,\epsilon},\varphi_{2,\epsilon})\|_{q,\infty}} \geq \frac{\int_0^M \|{\bf f}(t)\|_p{\rm d}t - \epsilon}{1+ O(\epsilon)}.
\end{equation}

We complete the proof by letting $\epsilon\rightarrow 0$.

{\bf Item} \ref{Item:Dirac}: Similarly to the previous part, for any smooth curve $\boldsymbol{\varphi}=(\varphi_1,\varphi_2)\in \mathcal{S}(\mathbb{T}_M)^2$ with $\|\boldsymbol{\varphi}\|_{q,\infty}=1$, we have that 
\begin{align}
\langle w_1, \varphi_1 \rangle +\langle w_2, \varphi_2 \rangle &= \sum_{k=0}^{K-1} \left(a_1[k] \varphi_1(t_k)+ a_2[k] \varphi_2(t_k)\right) \\& \leq \sum_{k=0}^{K-1} \|{\bf a}[k]\|_p \|\boldsymbol{\varphi}(t_k)\|_q \leq \sum_{k=0}^{K-1} \|{\bf a}[k]\|_p .
\end{align}

Taking the supremum over $\boldsymbol{\varphi}=(\varphi_1,\varphi_2)$  with $\|\boldsymbol{\varphi}\|_{q,\infty}=1$ then yields that $\|{\bf w}\|_{{\rm TV}-\ell_p} \leq \sum_{k=0}^{K-1} \|{\bf a}[k]\|_p$. To prove the equality, we first  define a set of vectors $\boldsymbol{\varphi}_k\in \mathbb{R}^2$  such that  $\|\boldsymbol{\varphi}_k\|_{\infty}=1$ and  ${\bf a}[k]^T \boldsymbol{\varphi}_k =\|{\bf a}[k]\|_p$ for $k=0, \ldots, K-1$. We then consider a smooth curve $\boldsymbol{\varphi}^*\in\mathcal{S}(\mathbb{T}_M)^2$ with $\|\boldsymbol{\varphi}^*\|_{q,\infty}=1$ such that $\boldsymbol{\varphi}^*(t_k) = \boldsymbol{\varphi}_k$. Using this, we then verify that 
$$\|{\bf w}\|_{{\rm TV}-\ell_p}\geq \langle w_1, \varphi_1^* \rangle +\langle w_2, \varphi_2^*\rangle  = \sum_{k\in\mathbb{Z}} \|{\bf a}[k]\|_p.$$
\end{proof}

\section{Proof of Proposition \ref{Prop:Invariance_L2}}\label{sec:appendix_invariance}

\begin{proof}
By substitution of $\textbf{R}_{\theta}\textbf{f}$ in \eqref{Eq:Reg}, we have that  

\begin{align}\label{eq:proof_invariance_1}
 \mathcal{R}(\textbf{R}_{\theta}\textbf{f}) 
    &= \nonumber \sup_{\substack{\boldsymbol{\varphi}\in \mathcal{S}(\mathbb{T}_M)^2\\ \|\boldsymbol{\varphi}\|_{2,\infty} = 1}} &&\mathbin{\stackunder[3pt]{\left( \langle \cos(\theta) f_1 - \sin(\theta) f_2, \varphi_1 \rangle \right.}{\left. + \ \langle \sin(\theta) f_1 + \cos(\theta) f_2, \varphi_2 \rangle \right)}}\\
    &= \sup_{\substack{\boldsymbol{\varphi}\in \mathcal{S}(\mathbb{T}_M)^2\\ \|\boldsymbol{\varphi}\|_{2,\infty} = 1}} &&\mathbin{\stackunder[3pt]{\left( \langle f_1, \cos(\theta)\varphi_1 + \sin(\theta)\varphi_2 \rangle \right.}{\left. + \ \langle f_2, -\sin(\theta) \varphi_1 + \cos(\theta) \varphi_2 \rangle \right)}}.
\end{align}

We perform the change of variable $\boldsymbol{\psi} = \textbf{R}_{-\theta} \boldsymbol{\varphi}$. We readily conclude that, since $\textbf{R}_{-\theta}$ is bijective over $\mathcal{S}(\mathbb{T}_M)^2$, for any $\boldsymbol{\varphi}\in \mathcal{S}(\mathbb{T}_M)^2$, we have that $\boldsymbol{\psi} = \textbf{R}_{-\theta} \boldsymbol{\varphi} \in \mathcal{S}(\mathbb{T}_M)^2$. Additionally, and in accordance with \eqref{Eq:L2LinfMixed}, we have that

\begin{equation}\label{eq:proof_invariance_2}
    \|\boldsymbol{\psi}\|_{2,\infty} = \sup_{t\in\mathbb{T}_M}\|\boldsymbol{\psi}(t)\|_2 = \sup_{t\in\mathbb{T}_M}\|\textbf{R}_{-\theta} \boldsymbol{\varphi}(t)\|_2 = \sup_{t\in\mathbb{T}_M}\|\boldsymbol{\varphi}(t)\|_2,
\end{equation}

\noindent as $\textbf{R}_{-\theta}$ is an isometry. Hence, it does not change the $\ell_2$ norm of a vector. Consequently, we have that

\begin{align}\label{eq:proof_invariance_3}
    \mathcal{R}(\textbf{R}_{\theta}\textbf{f}) 
    &= \nonumber \sup_{\substack{\boldsymbol{\psi}\in \mathcal{S}(\mathbb{T}_M)^2\\ \| \boldsymbol{\psi}\|_{2,\infty} = 1}} &&\left( \langle f_1,\psi_1\rangle + \langle f_2, \psi_2 \rangle\right)\\
    &= \mathcal{R}(\textbf{f}).
\end{align}

Moreover, according to Item \ref{Item:L1} of Theorem \ref{Thm:Intuition}, for any curve $\textbf{f} = (f_1,f_2)$ with absolutely integrable components $f_i \in L_1(\mathbb{T}_M), \ i= 1,2$, the $\mathrm{TV}-\ell_p$ norm is

\begin{equation}
     \norm{\textbf{f}}_{\mathrm{TV}-\ell_p} = \int_0^M (|f_1(t)|^p + |f_2(t)|^p)^{\frac{1}{p}} \ \mathrm{d}t.
\end{equation}

We take $f_1(t) = 1$, $f_2(t) = 0$, and $\theta = \frac{\pi}{4}$. This gives us

\begin{align}\label{eq:unrotated_lp}
       \norm{\textbf{f}}_{\mathrm{TV}-\ell_p}= 
      \int_0^M (|1|^p + |0|^p)^{\frac{1}{p}} \ \mathrm{d}t = M. 
\end{align}
 
When applying the planar rotation $\textbf{R}_{\theta}$ to the curve $\textbf{f}$, we have that
 
\begin{align}
     \norm{\textbf{R}_{\theta}\textbf{f}}_{\mathrm{TV}-\ell_p} &= 
      \nonumber \int_0^M (|f_1(t) \cos{(\theta)} - f_2(t) \sin{(\theta)}|^p \\
      \nonumber &\quad \quad \quad + |f_1(t) \sin{(\theta)} + f_2(t) \cos{(\theta)}|^p)^{\frac{1}{p}} \ \mathrm{d}t \\
      \nonumber &=\int_0^M (|\cos{(\theta)}|^p + |\sin{(\theta)}|^p)^{\frac{1}{p}} \ \mathrm{d}t \\
       &=\int_0^M \left ( 2 \left|\frac{\sqrt{2}}{2}\right|^p \right ) ^{\frac{1}{p}} \ \mathrm{d}t = 2^{\frac{1}{p} - \frac{1}{2}} M. \label{eq:rotated_lp}
\end{align}
 
We conclude that $\norm{\textbf{f}}_{\mathrm{TV}-\ell_p} = \norm{\textbf{R}_{\theta}\textbf{f}}_{\mathrm{TV}-\ell_p}$ if and only if $p = 2$, which proves that the $\mathrm{TV}-\ell_p$ norm is not rotation invariant for $p \neq 2$.
\end{proof}

\section{Representer Theorem}\label{sec:appendix_rep_theorem}
We start by providing the necessary tools before going into the proof of Theorem \ref{Thm:Rep} (see \cite{fageot2020tv} for more details). Specifically, we first describe the topological structure of the search space $\mathcal{X}_{\rm L}(\mathbb{T}_M)$. We then identify the set of extreme points of the RI-TV unit ball. Finally, we provide a full characterization of the solution set $\mathcal{V}$, from which we conclude Theorem \ref{Thm:Rep}.

\subsection{Search Space}

The space of periodic finite Radon measures is denoted by $\mathcal{M}(\mathbb{T}_M)$. It is a Banach space equipped with the total-variation norm

\begin{equation}
    \|w\|_{\rm TV} \overset{\Delta}{=} \sup_{\substack{\varphi\in \mathcal{S}(\mathbb{T}_M)\\ \|\varphi\|_{\infty}=1}} \langle w, \varphi\rangle.
\end{equation}

Subsequently, the native space associated to the pair $\left({\rm L}, \mathcal{M}(\mathbb{T}_M)\right)$  is defined as  $\mathcal{M}_{{\rm L}}(\mathbb{T}_M)=\{f\in \mathcal{S}'(\mathbb{T}_M): \|{\rm L}\{f\}\|_{\rm TV}<+\infty \}$. It has been shown that $\mathcal{M}_{{\rm L}}(\mathbb{T}_M)$ is isometrically isomorphic to $\mathcal{M}_0(\mathbb{T}_M) \times \mathbb{R}$, where $\mathcal{M}_0(\mathbb{T}_M)=\{ w\in \mathcal{M}(\mathbb{T}_M): \langle w, 1 \rangle=0\}$ is the space of  Radon measures with zero mean. The explicit form of such an isometry between spaces (and its inverse) is given by

\begin{align}
& \mathcal{M}_{{\rm L}}(\mathbb{T}_M)\rightarrow\mathcal{M}_0(\mathbb{T}_M) \times \mathbb{R}: f\mapsto \left({\rm L}\{f\}, \langle f,1\rangle \right), \nonumber \\
& \mathcal{M}_0(\mathbb{T}_M) \times \mathbb{R} \rightarrow \mathcal{M}_{{\rm L}}(\mathbb{T}_M):  \left(w, a\right) \mapsto {\rm L}^\dag \{w\}+ a, \label{eq:iso_inv}
\end{align}

where ${\rm L}^\dag$ is the pseudoinverse of ${\rm L}$. Finally, we note that the Green's function of ${\rm L}={\rm D}^{(\alpha+1)}$,  defined as $g_{\rm L} = {\rm L}^\dag \{\Sha\}$, is a continuous periodic function for all integers $\alpha \geq 1$ \cite{fageot2020tv}.   

We are now ready to characterize the topological structure of the  search space $\mathcal{X}_{\rm L}(\mathbb{T}_M)$ defined in \eqref{Eq:DefNative}.    

\footnotetext{\url{https://en.wikipedia.org/wiki/Daft_Punk}}

\begin{proposition}\label{Prop:Native}
The search space $\mathcal{X}_{\rm L}(\mathbb{T}_M)$ can be expressed as 
\begin{equation}\label{eq:equivalence}
\mathcal{X}_{\rm L}(\mathbb{T}_M)= \mathcal{M}_{{\rm L}}(\mathbb{T}_M)\times \mathcal{M}_{{\rm L}}(\mathbb{T}_M).
\end{equation}

Moreover, the mapping 
\begin{align}
&T_{\rm L}:\mathcal{X}_{{\rm L}}(\mathbb{T}_M)\rightarrow\mathcal{M}_0(\mathbb{T}_M)^2 \times \mathbb{R}^2 \nonumber \\ 
& T_{\rm L}({\bf r}) = \left({\rm L}\{r_1\},{\rm L}\{r_2\},  \langle r_1,1\rangle, \langle r_2,1\rangle  \right) 
\end{align}

is an isomorphism between $\mathcal{X}_{\rm L}(\mathbb{T}_M)$ and $\mathcal{M}_0(\mathbb{T}_M)^2 \times \mathbb{R}^2$ whose inverse is

\begin{align}
&T^{-1}_{\rm L}: \mathcal{M}_0(\mathbb{T}_M)^2 \times \mathbb{R}^2 \rightarrow \mathcal{M}_{{\rm L}}(\mathbb{T}_M) \nonumber \\
& T^{-1}_{\rm L}\left({\bf w}, {\bf a}\right) = \left({\rm L}^\dag \{w_1\}+ a_1,{\rm L}^\dag \{w_2\}+ a_2\right). 
\end{align}
\end{proposition}

\begin{proof}
Let ${\bf r}=(r_1,r_2)\in\mathcal{X}_{\rm L}(\mathbb{T}_M)$. We have that

\begin{align}
\mathcal{R}\left({\rm L}\{{\bf r}\}\right) & = \sup_{\substack{\boldsymbol{\varphi}\in \mathcal{S}(\mathbb{T}_M)^2\\ \|\boldsymbol{\varphi}\|_{2,\infty} = 1}} \left( \langle {\rm L}\{  r_1\},\varphi_1\rangle + \langle {\rm L} \{r_2\}, \varphi_2 \rangle\right) 
\nonumber\\& \geq \sup_{\substack{{\varphi}_1\in \mathcal{S}(\mathbb{T}_M)\\ \|(\varphi_1,0)\|_{2,\infty} = 1}} \langle {\rm L} \{r_1\},\varphi_1\rangle = \sup_{\substack{{\varphi}_1\in \mathcal{S}(\mathbb{T}_M)\\ \|\varphi_1\|_{\infty} = 1}}  \langle {\rm L} \{r_1\},\varphi_1\rangle\nonumber \\&= \|{\rm L}\{r_1\}\|_{\rm TV}, \label{eq:prodRI}
\end{align}

from which we deduce that $r_1 \in \mathcal{M}_{{\rm L}}(\mathbb{T}_M)$. Similarly, we get that $r_2 \in \mathcal{M}_{{\rm L}}(\mathbb{T}_M)$ and, hence, we have that $\mathcal{X}_{\rm L}(\mathbb{T}_M) \subseteq \left(\mathcal{M}_{{\rm L}}(\mathbb{T}_M)\right)^2$. For the reverse inclusion, let $r_1,r_2 \in \mathcal{M}_{\rm L}(\mathbb{T}_M)$. Using the inequalities $\|\boldsymbol{\varphi}\|_{2,\infty} \geq \|\varphi_i\|_{\infty}$ for $i=1,2$, we deduce that 

\begin{equation}
\left|\langle {\rm L}\{r_i\}, \varphi_i  \rangle \right| \leq \|{\rm L} \{r_i\}\|_{\rm TV} \|\varphi_i\|_\infty \leq  \|{\rm L}\{ r_i\}\|_{\rm TV}  \|\boldsymbol{\varphi}\|_{2,\infty}.
\end{equation}

Hence, we have that

\begin{align}
&\langle {\rm L}\{r_1\}, \varphi_1  \rangle +\langle {\rm L}\{r_2\}, \varphi_2  \rangle  \leq \nonumber \\ 
& \quad \left(\|{\rm L} \{r_1\}\|_{\rm TV} +\|{\rm L} \{r_2\}\|_{\rm TV}  \right) \|\boldsymbol{\varphi}\|_{2,\infty},
\end{align}

which  implies that 

\begin{equation}\label{eq:RIprod}
\mathcal{R}\left({\rm L}\{{\bf r}\}\right) \leq \|{\rm L} \{r_1\}\|_{\rm TV} +\|{\rm L} \{r_2\}\|_{\rm TV}  < +\infty.
\end{equation}

Hence, we have the inclusion ${\bf r}\in\mathcal{X}_{\rm L}(\mathbb{T}_M)$. 

Following \eqref{eq:RIprod} and \eqref{eq:prodRI}, we deduce that the norm topology of $\mathcal{X}_{\rm L}(\mathbb{T}_M)$ is equivalent to the product topology induced from $\mathcal{M}_{\rm L}(\mathbb{T}_M)\times \mathcal{M}_{\rm L}(\mathbb{T}_M)$. This, together with the fact that $\mathcal{M}_{\rm L}(\mathbb{T}_M)$ is isometrically isomorphic to $\mathcal{M}_0(\mathbb{T}_M)\times \mathbb{R}$, implies that $T_{\rm L}$ is an isomorphism. Its inverse is readily deduced from \eqref{eq:iso_inv}.  
\end{proof}

\subsection{Extreme Points of the RI-TV Unit Ball}
Our strategy to characterize the solution set $\mathcal{V}$ defined in \eqref{eq:continuous_problem} consists of invoking the main result of Boyer {\it et al.} \cite{boyer2019representer}, which requires the knowledge of the form of extreme points of the unit ball of the regularization functional. To that end, we prove that the extreme points of the RI-TV unit ball are vector-valued Dirac combs. 

\begin{proposition}\label{Prop:EP_RITV}
 An element ${\bf w}^*\in \mathcal{M}(\mathbb{T}_M)^2$ is an extreme point of the RI-TV unit ball $B=\{ {\bf w}\in \mathcal{M}(\mathbb{T}_M)^2: \mathcal{R}({\bf w}) =1\}$ if and only if it is a vector-valued Dirac comb of the form ${\bf w}^* = {\bf a}\Sha_M(\cdot -t_0)$ for some $t_0\in\mathbb{T}_M$ and ${\bf a}\in\mathbb{R}^2$ with $\|{\bf a}\|_2=1$.
\end{proposition}

\begin{proof}
Assume by contradiction that there exists an extreme point ${\bf w}^*$ of $B$ that is not a Dirac comb. This implies that  there exists an interval $I\subseteq \mathbb{T}_M$ such that ${\bf w}_1={\bf w}^* \mathbbm{1}_I$ and ${\bf w}_2={\bf w}^* \mathbbm{1}_{I^c}$ are both nonzero Radon measures that satisfy ${\bf w}^* = {\bf w}_1 + {\bf w}_2$. Due to their disjoint support, we have that $\mathcal{R}({\bf w}^*)= \mathcal{R}(\mathbf{w}_1)+ \mathcal{R}(\mathbf{w}_2)$. Let us now define the measures 
${\bf w}_+ = (1+\epsilon) {\bf w}_1 + (1-\delta){\bf w}_2$ and ${\bf w}_- = (1-\epsilon) 
{\bf w}_1 + (1+\delta){\bf w}_2$, where $\epsilon,\delta>0$ are small constants such that 
$\epsilon \mathcal{R}({\bf w}_1) = \delta \mathcal{R}({\bf w}_2)$. By observing  that $\mathcal{R}({\bf 
w}_{+}) = \mathcal{R}({\bf w}_{-}) = 1$ and ${\bf w}^*= \frac{{\bf w}_{+}+{\bf w}_{-}}{2}$, we conclude that ${\bf w}^*$ is not an extreme point of $B$, which yields a contradiction. Hence, the  extreme points of $B$
can only be vector-valued Dirac combs. 

To prove the reverse inclusion, let ${\bf w}^* ={\bf a}\Sha_M(\cdot -t_0)$ with $\|{\bf a}\|_2=1$. We now prove that ${\bf w}^*$ is an extreme point of $B$. Assume that there exist ${\bf w}_1, {\bf w}_2\in B$ such that ${\bf w}^*= \frac{1}{2} ({\bf w}_1 + {\bf w}_2)$. Let us define the measure ${\bf w}_0= {\bf w}_1 \mathbbm{1}_{t\neq t_0}\in\mathcal{M}(\mathbb{T}_M)^2$ so that ${\bf w}_1= {\bf w}_0 + {\bf a}_1\Sha_M(\cdot -t_0)$ for some ${\bf a}_1\in\mathbb{R}^2$. We then must have ${\bf w}_2= (-{\bf w}_0) + {\bf a}_2\Sha_M(\cdot -t_0)$ with ${\bf a}=\frac{1}{2}({\bf a}_1+{\bf a}_2)$. The construction implies that

\begin{equation}
    1=\mathcal{R}({\bf w}_i) = \mathcal{R}({\bf w}_0) + \|{\bf a}_i\|_2 \geq \|{\bf a}_i\|_2, i = 1, 2.
\end{equation}

This, together with the triangle inequality, yields

\begin{equation}
    2= \|2{\bf a}\|_2 \leq \|{\bf a}_1\|_2 +\|{\bf a}_2\|_2 \leq 1+ 1 =2.
\end{equation}

Hence, all inequalities must be saturated. In particular, we must have that ${\bf w}_0=\boldsymbol{0}$ and $\|{\bf a}\|_2= \frac{1}{2}(\|{\bf a}_1\|_2+\|{\bf a}_2\|_2)$. Finally, we invoke the strict convexity of the $\ell_2$ norm to conclude that ${\bf a}={\bf a}_1={\bf a}_2$ and, thus, that ${\bf w}_1 = {\bf w}_2$, which in turn implies that ${\bf w}^*$ is an extreme point of $B$.
\end{proof}

\subsection{Representer Theorem}
We now provide a complete characterization of the solution set $
\mathcal{V}$ in \eqref{eq:continuous_problem} from which we readily deduce Theorem \ref{Thm:Rep} as a corollary. 
\begin{theorem}
  The solution set \eqref{eq:continuous_problem} is nonempty, convex, and weak*-compact. Moreover, any extreme point ${\bf r}^*$ of  $\mathcal{V}$ is a periodic ${\rm L}$-spline that satisfies   \eqref{eq:extreme_point}.
\end{theorem}
\begin{proof}
Let us define the cost functional $E:\mathcal{M}(\mathbb{T}_M)^2 \times \mathbb{R}^2\rightarrow \mathbb{R}\cup\{+\infty\}$ as 

\begin{equation}
E({\bf w}, {\bf a})= \sum_{m=0}^{M-1} \norm{\boldsymbol{\nu}_m({\bf w}) + {\bf a} - \textbf{p}[m]}_2^2    +\sum_{i=1}^2 \chi_{\langle w_i,1 \rangle = \boldsymbol{0}},
\end{equation}
where $\chi_A$ denotes the characteristic function of the set $A$, and  $\boldsymbol{\nu}_m= (\nu_{m,1}, \nu_{m,2})$ with 

\begin{align}
\nu_{m,i}({\bf w})&= \left.\left( {\rm L}^\dag\{w_i\}(t)\right)\right\vert_{t=m} = \langle {\rm L}^\dag  \{w_i\}, \Sha(\cdot -m)\rangle \nonumber \\
&= \langle w_i, {\rm L}^{\dag *}  \{\Sha(\cdot -m)\}\rangle = \langle w_i, g_{\rm L}(m-\cdot)\rangle 
\end{align}
for $i=1,2$. We note that $\boldsymbol{\nu}_m$ is weak*-continuous in the topology of $\mathcal{M}_{\rm L}(\mathbb{T}_M)$ due to the inclusion $g_{\rm L}\in \mathcal{C}(\mathbb{T}_M)$. 

Then, we formulate a minimization problem that admits the solution set
\begin{align}\label{eq:eqv_min}
    \tilde{\mathcal{V}}=\argmin_{\substack{{\bf w} \in \mathcal{M}(\mathbb{T}_M)^2 \\ {\bf a}\in \mathbb{R}^2}}   \left( \mathcal{J}({\bf w}, {\bf a}) =  E({\bf w}, {\bf a}) +  \lambda \mathcal{R}({\bf w}) \right).
\end{align}
Following the general representer theorem of Unser and Aziznejad for the minimization of seminorms \cite[Theorem 3]{unser2021convex}, we deduce that the solution set $
\tilde{V}$ is nonempty, convex, and weak*-compact. Moreover, any extreme point  of $\tilde{\mathcal{V}}$ can be written as $({\bf w}^*,{\bf a}^*)$, where ${\bf w}^*= \sum_{k=0}^{K-1} {\bf a}_k \Sha(\cdot-t_k)$ with $K\leq 2M$   for some ${\bf a}_k \in\mathbb{R}^2$ and $t_k\in\mathbb{T}_M$. 

The final step is observe that the isomorphism $T_{\rm L}$ defined in Proposition \ref{Prop:Native} allows us to state that

\begin{equation}
    E(T_{\rm L}({\bf r})) =\sum_{m=0}^{M-1} \norm{\left.\textbf{r}(t)\right\vert_{t=m} - \textbf{p}[m]}_2^2
\end{equation}
for any ${\bf r}\in \mathcal{X}_{\rm L}(\mathbb{T}_M)$, from which we conclude that $\tilde{\mathcal{V}}=T_{\rm L}\left(\mathcal{V}\right)$. Hence, the solution set $\mathcal{V}= T^{-1}_{\rm L}(\tilde{\mathcal{V}})$ is  nonempty, convex, and weak*-compact, and any  extreme point ${\bf r}^*$ of $\mathcal{V}$ induces an extreme point $({\bf w}^*,{\bf a}^*)= T_{\rm L}({\bf r}^*)$ of $\tilde{V}$. In particular, we have that

\begin{equation}
    {\rm L}\{{\bf r}^*\} = {\bf w}^* =\sum_{k=0}^{K-1} {\bf a}_k \Sha(\cdot-t_k), \quad K\leq 2M.
\end{equation}

\end{proof} 

\section{Proof of Theorem \ref{Thm:Hybrid}} \label{sec:app_thmhyb}
 \begin{proof}
By invoking Proposition \ref{Prop:Native},  we deduce that there is a bijection between $\mathcal{V}_{\rm hyb}$ and the solution set 
\begin{equation}\label{pb:hyb_red}
    \tilde{\mathcal{V}}_{\rm hyb} = \argmin_{\substack{{\bf w}_1,{\bf w}_2\in \mathcal{M}_0(\mathbb{T})^2\\{\bf a}_1, {\bf a}_2 \in \mathbb{R}^2}}  E\left( {\bf w}_1,  {\bf a}_1,   {\bf w}_2,  {\bf a}_2\right)+  \lambda_1 \mathcal{R}({\bf w}_1) + \lambda_2 \mathcal{R}({\bf w}_2) ,
\end{equation}
where the data fidelity cost $E: \left( \mathcal{M}(\mathbb{T})^2\times  \mathbb{R}^2\right)^2 \rightarrow\mathbb{R}_{\geq 0}$ satisfies 
\begin{equation}
E( T_{{\rm L}_1}({\bf r}_1) , T_{{\rm L}_2}({\bf r}_2) ) = \sum_{m=0}^{M-1} \norm{\left.\textbf{r}_1(t)\right\vert_{t=m} + \left.\textbf{r}_2(t)\right\vert_{t=m} - \textbf{p}[m]}_2^2. 
\end{equation}
This implies that there is a bijection between $\mathcal{V}_{\rm hyb}$ and $\tilde{\mathcal{V}}_{\rm hyb}$.  The last step is to note that for any extreme point  $({\bf w}_1^*,{\bf w}_2^*)$ of the unit ball $\{({\bf w}_1,{\bf w}_2)\in \mathcal{M}(\mathbb{T})^4: \lambda_1 \mathcal{R}({\bf w}_1) + \lambda_2 \mathcal{R}({\bf w}_2)\leq 1 \}$,  we have that ${\bf w}_1^*=\boldsymbol{0}$ or ${\bf w}_2^*= \boldsymbol{0}$.  This together with \cite[Theorem 3 ]{unser2021convex} concludes the proof. 
\end{proof} 
\section*{Acknowledgments}

The authors would like to thank Julien Fageot for his useful comments on the manuscript as well as the fruitful discussions deriving from them.

\bibliography{bibtex/bib/IEEEabrv.bib,bibtex/bib/bibliography.bib}{}
\bibliographystyle{IEEEtran}

\end{document}